\documentclass[aip,graphicx,reprint,onecolumn,11pt]{revtex4-1}
\usepackage{amsmath,amsthm,amssymb}
\usepackage{amsfonts}
\usepackage{amssymb}
\usepackage{mathptmx}
\usepackage{color}
\usepackage{braket}

\evensidemargin=\oddsidemargin
\newtheorem{theorem}{Theorem}[section]
\newtheorem{proposition}[theorem]{Proposition}
\newtheorem{lemma}[theorem]{Lemma}
\newtheorem{corollary}[theorem]{Corollary}

\theoremstyle{definition}
\newtheorem{definition}{Definition}[section]
\newtheorem{example}{Example}[section]

\newtheorem{idea memo}{Idea Memo}[section]

\theoremstyle{remark}
\newtheorem*{remark}{Remark}

  \newcommand{\al}{\alpha}
  
  \newcommand{\ga}{\gamma}

  \newcommand{\ch}{\chi}
 
  \newcommand{\ka}{\kappa}
  \newcommand{\la}{\lambda}
  
  \newcommand{\ph}{\phi}
 
  \newcommand{\rh}{\rho}
  \newcommand{\si}{\sigma}

  \newcommand{\vp}{\varphi}

  \newcommand{\Ga}{\Gamma}
  \newcommand{\De}{\Delta}

  \newcommand{\Ph}{\Phi}
  \newcommand{\Ps}{\Psi}

\newcommand{\red}{\textcolor{red}}

\newcommand{\cut}[1]{}
\newcommand{\II}{I}
\newcommand{\alg}{\mathrm{alg}}
\newcommand{\bin}{\mathrm {bin}}
\newcommand{\nor}{\mathrm {nor}}
\newcommand{\Rep}{\mathrm{Rep}}
\newcommand{\bx}{\mathbf{x}}
\newcommand{\by}{\mathbf{y}}
\newcommand{\bA}{\mathbf{A}}

\newcommand{\bS}{\mathbf{S}}
  
\newcommand{\bbM}{\M}
\newcommand{\cA}{\mathcal{A}}
\newcommand{\cB}{\mathcal{B}}
\newcommand{\cE}{\mathcal{E}}
\newcommand{\cF}{\mathcal{F}}
\newcommand{\cH}{\mathcal {H}}
\newcommand{\cI}{\mathcal {I}}
\newcommand{\cK}{\mathcal{K}}
\newcommand{\cL}{\mathcal{L}}
\newcommand{\cM}{\mathcal{M}}
\newcommand{\cN}{\mathcal{N}}
\newcommand{\cO}{\mathcal{O}}
\newcommand{\cX}{\mathcal{X}}
\newcommand{\cY}{\mathcal{Y}}
\newcommand{\cS}{\mathcal{S}}
\newcommand{\tT}{\widetilde{T}}
\newcommand{\id}{{\rm id}}

\newcommand{\beq}{\begin{equation}}
\newcommand{\beqa}{\begin{eqnarray}}
\newcommand{\beqas}{\begin{eqnarray*}}
\newcommand{\beql}[1]{\begin{equation}\label{eq:#1}}
  \newcommand{\eq}[1]{(\ref{eq:#1})}
    \newcommand{\Eq}[1]{Eq.~(\ref{eq:#1})}
\newcommand{\eeq}{\end{equation}}
\newcommand{\eeqa}{\end{eqnarray}}
\newcommand{\eeqas}{\end{eqnarray*}}
\newcommand{\SM}{\cS_n(\cM)}
\newcommand{\bracket}{\braket}
\newcommand{\Ir}{\|\cI\rh\|}

\renewcommand{\ae}{a.e.\ }
\newcommand{\BH}{\mathcal{B}(\mathcal{H})}
\newcommand{\B}{\mathcal{B}}
\newcommand{\CP}{{\mbox{CP}}}
\newcommand{\R}{\mathbb{R}}
\newcommand{\M}{\mathbb{M}}
\newcommand{\aotimes}{\otimes_{\alg}}
\newcommand{\botimes}{\otimes_{\bin}}
\newcommand{\motimes}{\otimes_{\min}}
\newcommand{\ootimes}{\overline{\otimes}}
\newcommand{\os}{\widetilde{ [s]}}
\newcommand{\Cite}[1]{Ref.~\onlinecite{#1}}
\newcommand{\Cites}[1]{Refs.~\onlinecite{#1}}

\newcommand{\ba}{\mathbf{a}}
\newcommand{\bb}{\mathbf{b}}

\begin{document}

\title{\Large Measurement theory in local quantum physics}
\author{Kazuya Okamura}
\email{okamura@math.cm.is.nagoya-u.ac.jp}
\author{Masanao Ozawa}
\email{ozawa@is.nagoya-u.ac.jp}
\affiliation{${}^1$Graduate School of Information Science, Nagoya University,\\
Chikusa-ku, Nagoya 464-8601, Japan }
\date{}
\begin{abstract}
In this paper, we aim to establish foundations of measurement theory in local quantum physics.
For this purpose, we discuss a representation theory of completely positive (CP) instruments 
on arbitrary von Neumann algebras.  We introduce a condition called the normal extension property (NEP) and establish a one-to-one correspondence between CP instruments with  the NEP and statistical equivalence classes of measuring processes.  We show that every CP instrument on an atomic von Neumann algebra has the NEP, extending the well-known result for type I factors.
Moreover, we show that every CP instrument on an injective von Neumann algebra is approximated by CP instruments with the NEP.  The concept of posterior states is also discussed to show that the NEP is equivalent to the existence of a strongly measurable family of posterior states for every normal state.
Two examples of CP instruments without the NEP are obtained from this result.
It is thus concluded that in local quantum physics not every CP instrument represents 
a measuring process, but in most of physically relevant cases every CP instrument 
can be realized by a measuring process within arbitrary error limits, as every
approximately finite dimensional (AFD) von Neumann algebra on a separable Hilbert space is injective.
To conclude the paper, the concept of local measurement in algebraic quantum field theory is 
examined in our framework.  In the setting of the  Doplicher-Haag-Roberts 
and Doplicher-Roberts (DHR-DR) theory describing local excitations, 
we show that an instrument on a local algebra can be extended to a local instrument on 
the global algebra if and only if it is a CP instrument with the NEP, provided that the split 
property holds for the net of local  algebras.
\end{abstract}

\maketitle

\section{Introduction}
\label{se:introduction}
This paper represents the first step of our attempt towards establishing 
measurement theory in local quantum physics \cite{Haa96}.
Quantum measurement theory is an indispensable part of quantum theory, 
which was demanded for
foundations of quantum theory \cite{vN32} and 
provided 
a theoretical basis for quantum information technology \cite{NC00}.
In particular, mathematical theory of quantum measurements
established in the 1970s and the 1980s has made a great success 
in revealing our ability of making precision measurements much broader 
than what was assumed in the conventional approach established in the 1930s,
as shown in the resolution of a dispute about the sensitivity limit for gravitational 
wave detectors \cite{BVT80,Yue83,Cav85,88MS,89RS,Mad88},
and the derivations of universally-valid measurement uncertainty relations
\cite{Hei27,App98c,03UVR,03HUR,03UPQ,04URJ,04URN,Hal04,Wer04,Bra13,BLW13,14NDQ,BLW14JMP,Bra14}
with their experimental demonstrations  
\cite{12EDU,RDMHSS12,13EVR,13VHE,RBBFBW14,14ETE,15A1}.

Mathematical study of quantum measurement began with the famous book 
\cite{vN32} written by von Neumann, based on the so-called repeatability hypothesis.  
Nakamura and Umegaki \cite{NU62} attempted to generalize von Neumann's theory
to continuous observables by the mathematical concept of conditional expectation \cite{Ume54}.
Arveson \cite{Arv67}, however, later showed a no-go theorem for their approach.
Davies and Lewis \cite{DL70} proposed abandoning the repeatability hypothesis
to develop a more flexible approach to quantum measurement theory.
To this end, they introduced the mathematical concept of instrument \cite{DL70,Dav76}
as a framework to analyze statistical properties of general quantum measurements 
which do not necessarily satisfy the repeatability hypothesis, 
extending the notions of operation introduced, for instance,
by Schwinger \cite{Sch59,*Sch60b} and Haag and Kastler \cite{HK64}
as well as effects introduced, for instance, by Ludwig \cite{Lud67,*Lud68}.
At almost the same time, Kraus \cite{Kra71,Kra83} introduced complete positivity 
in the concept of operation and studied measurement processes of yes-no measurements.  
Following those studies, one of the present authors introduced complete positivity
in the concept of instrument and showed that completely positive (CP) instruments
perfectly describe measuring processes in quantum mechanics up to statistical 
equivalence \cite{84QC}.  This result finalized the mathematical characterization
of general measurements in quantum mechanics
(see also Ref.~\onlinecite{04URN} for an axiomatic characterization of general
quantum measurements).

More specifically, the above result is based on a representation theorem of 
CP instruments stating that every CP instrument defined for a quantum system 
with finite degrees of freedom, algebraically represented by a type I factor, 
can be obtained from a measuring process specified by a unitary evolution of the
composite system with a measuring apparatus and by a subsequent direct measurement 
of a meter in the apparatus, and vice versa.
However, this theorem does not have a straightforward extension to arbitrary
von Neumann algebras, since the proof relies on the uniqueness of irreducible normal 
representations of a type I factor up to unitary equivalence \cite{Arv76}. 
Naturally, this difficulty is considered one of major obstacles in generalizing 
quantum measurement theory to quantum systems with infinite degrees of freedom.
In order to overcome this difficulty, here, we study possible extensions of 
the above representation theorem of CP instruments to general ($\si$-finite) 
von Neumann algebras and apply to quantum systems of infinite degrees of freedom
in the framework of algebraic quantum field theory \cite{Ara00}.

One of main results in the present paper is to give a necessary and sufficient
condition for a CP instrument to describe a physical process of measurement.
In mathematical description of a quantum measurement,
it is essential to consider both the von Neumann algebra $\cM$ of 
bounded observables of the measured system and the probability measure 
space $(S,\cF,\mu)$ describing the possible outcomes of measurement 
shown by the meter in the apparatus.
An essential role of CP instruments is to connect them.
From an algebraic point of view, the outcome of measurement is also described by
the abelian von Neumann algebra $L^{\infty}(S,\mu)$ of bounded random variables 
on $(S,\cF,\mu)$, so that it is natural to 
form a certain tensor product 
algebra of $\cM$ and $L^{\infty}(S,\mu)$.
Apart from their algebraic tensor product $\cM\aotimes L^{\infty}(S,\mu)$,
their C*-algebraic binormal tensor product $\cM\botimes L^{\infty}(S,\mu)$ 
arises naturally, since we, first of all,  shall show that every CP instrument $\cI$
can be uniquely extended to a unital CP map $\Psi_\cI$ of $\cM\botimes L^{\infty}(S,\mu)$ 
to $\cM$.  There are many different kinds of operator algebraic tensor product, but no simple
algebraic consideration can suggest what kind of tensor product is a relevant choice,
due to the lack of a general treatment for compositions of different systems in algebraic
quantum theory.  On the other hand, in the case where $\cM$ is a type I factor, 
it is known by the previous investigation that a CP instrument is uniquely extended 
to a unital normal CP map of the $W^*$-tensor product 
$\cM\ootimes L^{\infty}(S,\mu)$ to $\cM$.
In view of this case, it is natural to examine the extendability of the unital CP map 
$\Psi_\cI$ of $\cM\botimes L^{\infty}(S,\mu)$ to $\cM$ 
to a unital normal CP map $\widetilde{\Psi_\cI}$ of $\cM\ootimes L^{\infty}(S,\mu)$ 
to $\cM$.  We regard this extendability as a key property of CP instruments and call it the 
normal extension property (NEP).
Then, it is easily seen that every measuring process determines a CP instrument with  the NEP.
We shall prove the converse that every CP instrument with  the NEP has a corresponding
measuring process by applying a structure theorem of normal representations of von Neumann 
algebras to the Stinespring representation of $\widetilde{\Psi_\cI}$.
In this way, we can avoid the use of the uniqueness theorem of irreducible normal 
representations of a type I factor.
Therefore,  the NEP for CP instruments is a condition equivalent to the existence of the 
corresponding measuring processes.

It should be mentioned that usually it is not easy to check whether a given CP instrument 
has the NEP or not.  We consider the problem as to
how ubiquitously such CP instruments exist.   
As above, the set of instruments describing measuring processes is characterized 
by the set $\mathrm{CPInst}_{\mathrm{NE}}(\cM,S)$ of CP instruments with the NEP,
which is a subset of the set $\mathrm{CPInst}(\cM,S)$ of CP instruments.
In the case where $\cM$ is a type I factor,  it is known that 
$\mathrm{CPInst}_{\mathrm{NE}}(\cM,S)=\mathrm{CPInst}(\cM,S)$ holds \cite{84QC}.  
We shall show that this relation also holds if $\cM$ is a direct sum of type I factors, 
while the equality does not hold even for a type I von Neumann algebra as discussed below. 
Thus, the next problem is whether  $\mathrm{CPInst}_{\mathrm{NE}}(\cM,S)$ is
experimentally dense in $\mathrm{CPInst}(\cM,S)$ in the sense that every CP instrument
can be approximated by a CP instrument with the NEP within arbitrary error limits.
This problem is affirmatively solved for injective von Neumann algebras.
In most of physically relevant cases the algebras of local observables are known to be
injective, as they are separable approximately finite dimensional (AFD) 
von Neumann algebras, and hence this result will provide a satisfactory basis 
for measurement theory in local quantum physics.

In their seminal paper \cite{DL70} Davies and Lewis conjectured the non-existence 
of (weakly) repeatable instruments for continuous observables in the standard
formulation of quantum mechanics.
An instrument is called weakly repeatable if it satisfies an analogous condition with
von Neumann's repeatability hypothesis \cite{vN32}.
In Ref.~\onlinecite{85CA}, this conjecture is proved by connecting discreteness of
a weakly repeatable instrument with the existence of a family of posterior states, 
which determines the state just after the measurement given each individual 
value of measurement outcome.
In this paper, we shall prove that the NEP is equivalent to the existence of a strongly 
measurable family of posterior states for every normal state.
From this result, two examples of CP instruments on injective von Neumann algebras
without the NEP are obtained,
which arise from weakly repeatable CP instruments for continuous observables in
a commutative (type I) von Neumann algebra and a type II$_1$ factor.
Thus, in the general case there exists a weakly repeatable CP instrument for a
continuous observable that does not have the corresponding measuring process, 
whereas for a separable type I factor every CP instrument has the corresponding measuring 
process but no weakly repeatable instrument exists for continuous observables.

By making use of the NEP, we also develop measurement theory in algebraic quantum field 
theory (AQFT) and characterize local measurements under the Haag duality as postulated 
in the DHR-DR theory
\cite{DHR69a, *DHR69b,*DHR71,*DHR74,DR89a,*DR89b,*DR90} 
and the split property, which is derived, for instance, by the nuclearity
condition \cite{BW86}.  Under those assumptions,
we show that an instrument on a local algebra can be extended to a local instrument on 
the global algebra if and only if it is a CP instrument with the NEP.
Thus, we conclude that the experimental closure of the statistical equivalence classes 
of local measurements on a given spacetime region is represented by the set of CP interments
defined on the local algebra of that region.
As above, we overcome a difficulty in the previous investigations in mathematical 
characterizations of quantum measurement and open up local measurement theory 
in quantum systems of infinite degrees of freedom.

The necessary preliminaries are given in section \ref{se:preliminaries};
several tensor products of operator algebras and a structure theorem 
of normal representations of von Neumann algebras are summarized.
In section \ref{se:completely}, we discuss a representation theory of CP instruments and
establish a one-to-one correspondence between CP instruments
with  the NEP and statistical equivalence classes of measuring processes.
In section  \ref{se:approximations}, 
it is shown that every CP instrument on an atomic von Neumann algebra has  the NEP.
Moreover, we prove a density theorem stating that every CP instrument on an injective 
von Neumann algebra can be approximated by CP instruments with  the NEP.
Thus, we establish that the NEP holds approximately in most of physically relevant cases.
In section \ref{se:existence}, the existence problem of a family of posterior states is discussed. 
It is proved that  the NEP is equivalent to
the existence of a strongly measurable family of posterior states for every normal state.
From this result, two examples of CP instruments without  the NEP are obtained.
In section \ref{se:dhr-dr},  the concept of local measurements 
in algebraic quantum field theory is examined in our framework.
In the setting of the DHR-DR theory describing local excitations,
we show that any physically relevant local measurement carried out in a local spacetime 
region is represented by a CP instrument with NEP defined on the corresponding local 
algebra and conversely that every CP instrument defined on a local algebra represents 
a local measurement within arbitrary error limits.

\section{Preliminaries}
\label{se:preliminaries}
A representation of a C*-algebra  $\cX$ on a Hilbert space  $\cH$ is
a *-homomorphism of $\cX$ into the algebra $\B(\cH)$ 
of bounded linear operators on $\cH$.
Let $\cX$ and $\cY$ be C*-algebras and $\cH$ a Hilbert space.
We denote by $\Rep(\cX)$ the class of representations of $\cX$, by
$\Rep(\cX;\cH)$ the set of representations of $\cX$ on $\cH$,
and by $\mathbf{Hilb}$ the class of Hilbert spaces.
We define two norms $\|\cdot\|_{\min}$ and $\|\cdot\|_{\max}$
on the algebraic tensor product $\cX\otimes_{\alg}\cY$
of $\cX$ and $\cY$ by
\begin{align}
\| A\|_{\textrm{min}}&=\sup_{(\pi_1,\pi_2)\in\Rep(\cX)\times\Rep(\cY)}
\| \sum_{j=1}^n\pi_1(X_j)\otimes\pi_2(Y_j)\|, \\
\| A\|_{\textrm{max}}&=\sup_{(\pi_1,\pi_2)\in
\II_{\max}(\cX,\cY)}\| \sum_{j=1}^n\pi_1(X_j)\pi_2(Y_j)\|,
\end{align}
respectively,
for every $A=\sum_{j=1}^nX_j\otimes_{\alg}Y_j\in\cX\otimes_{\alg}\cY$,
where
\begin{align}
\II_{\max}(\cX,\cY)&=\bigcup_{\cH\in\mathbf{Hilb}}
\II_{\max}(\cX,\cY;\cH),\\
\II_{\max}(\cX,\cY;\cH) &= 
\{(\pi_1,\pi_2)\in\Rep(\cX;\cH)\times\Rep(\cY;\cH)
 |  \nonumber\\
&\hspace{15mm} \mbox{$[\pi_1(X),\pi_2(Y)]=0$ for all $X\in\cX$ and $Y\in\cY$}\}.
\end{align}
We call the completion
$\cX\motimes \cY$ ($\cX\otimes_{\max}\cY$, resp.)
of $\cX\otimes_{\alg}\cY$ with respect to
the norm $\|\cdot\|_{\min}$ ($\|\cdot\|_{\max}$, resp.)
the minimal (maximal, resp.) tensor product of $\cX$ and $\cY$.

Let $\cM$ be a von Neumann algebra and $\cY$ a C*-algebra.
We denote by $\Rep_n(\cM;\cH)$ the set of
normal representations of $\cM$ on $\cH$.
We call the completion $\cM\otimes_{\nor}\cY$
of $\cM\otimes_{\alg}\cY$ with respect to
the norm $\|\cdot\|_{\nor}$ defined below the normal tensor product
of $\cM$ and $\cY$:
\begin{equation}
\| X\|_{\textrm{nor}}=\sup_{(\pi_1,\pi_2)\in
\II_{\nor}(\cM,\cY)}\| \sum_{j=1}^n\pi_1(M_j)\pi_2(Y_j)\|
\end{equation}
for every $X=\sum_{j=1}^nM_j\otimes_{\alg}Y_j\in\cM\otimes_{\alg}\cY$,
where
\begin{align}
\II_{\nor}(\cM,\cY)&=
\bigcup_{\cH\in\mathbf{Hilb}}
\II_{\nor}(\cM,\cY;\cH),\\
\II_{\nor}(\cM,\cY;\cH) &= 
\{(\pi_1,\pi_2)\in\Rep_n(\cM;\cH)\times\Rep(\cY;\cH) |
  \nonumber\\
&\hspace{15mm} \mbox{$[\pi_1(M),\pi_2(Y)]=0$ for all $M\in\cM$ and $Y\in\cY$}\}.
\end{align}

Let $\cM$ and $\cN$ be von Neumann algebras.
We call the completion $\cM\botimes \cN$
of $\cM\otimes_{\alg}\cN$ with respect to
the norm $\|\cdot\|_{\bin}$ defined below
the binormal tensor product of $\cM$ and $\cN$:
\begin{equation}
\| X\|_{\textrm{bin}}=\sup_{(\pi_1,\pi_2)\in
\II_{\bin}(\cM,\cN)}\| \sum_{j=1}^n\pi_1(M_j)\pi_2(N_j)\|
\end{equation}
for every $X=\sum_{j=1}^nM_j\otimes_{\alg}N_j\in\cM\otimes_{\alg}\cN$,
where
\begin{align}
\II_{\bin}(\cM,\cN)&=
\bigcup_{\cH\in\mathbf{Hilb}}
\II_{\bin}(\cM,\cN;\cH),\\
\II_{\bin}(\cM,\cN;\cH) &= 
\{(\pi_1,\pi_2)\in\Rep_n(\cM;\cH)\times\Rep_n(\cN;\cH) |
  \nonumber\\
&\hspace{15mm}\mbox{$[\pi_1(M),\pi_2(N)]=0$ for all $M\in\cM$ and $N\in\cN$}\}.
\end{align}

The maximal tensor product $\cX\otimes_{\max}\cY$,
the normal tensor product $\cM\otimes_{\nor}\cY$,
and the binormal tensor product $\cM\botimes \cN$
have the following properties:

\begin{proposition}[\text{Ref.~\onlinecite[Chapter IV, Proposition 4.7]{Tak79}}]\label{Max}
Let $\cX$ and $\cY$ be C*-algebras.
Let $\cM$ and $\cN$ be W*-algebras.
Let $\cH$ be a Hilbert space.
For every $(\pi_1,\pi_2)\in\II_{\max}(\cX,\cY;\cH)$,
{\rm [}$(\pi_1,\pi_2)\in\II_{\nor}(\cM,\cY;\cH)$, or
$(\pi_1,\pi_2)\in\II_{\bin}(\cM,\cN;\cH)$,
resp.{\rm ]}
there exists a representation $\pi$ of $\cX\otimes_{\max}\cY$ 
{\rm [}$\cM\otimes_{\nor}\cY$, or
$\cM\botimes \cN$, resp.{\rm ]}
on $\cH$
such that
\begin{equation}
\pi(X\otimes Y)=\pi_1(X)\pi_2(Y)
\end{equation}
for all $X\in\cX$ and $Y\in\cY$
{\rm [}$X\in\cM$ and $Y\in\cY$,
or $X\in\cM$ and $Y\in\cN$, resp.{\rm]}.
\end{proposition}
 
A C*-algebra $\cX$ is said to be nuclear if
\begin{equation}
\cX\motimes \cY=\cX\otimes_{\max}\cY
\end{equation}
for every C*-algebra $\cY$ \cite{Tak02}.
It is known that C*-tensor products with nuclear C*-algebras are unique.
A C*-algebra $\cX$ on a Hilbert space $\cH$ is said to be injective if
there exists a norm one projection of $\B(\cH)$ onto $\cX$.
It is proven by Effros and Lance\cite{EL77} that
a von Neumann algebras $\cM$ is injective if and only if
\begin{equation}
\cM\motimes \cY=\cM\otimes_{\nor}\cY
\end{equation}
for every C*-algebra $\cY$.
Abelian C*-algebras are both nuclear and injective.
A characterization of von Neumann algebras which are nuclear as C*-algebras
is given in Brown-Ozawa \cite[Proposition 2.4.9]{BO08}. 

\begin{theorem}[Arveson \protect{\cite[Theorem 1.3.1]{Arv69},
Ref.~\onlinecite[Theorem 12.7]{Pau02}}]\label{CommLift}
Let $\cH$, $\cK$ be Hilbert spaces, and $\mathcal{B}$ a unital C*-subalgebra of
$\B(\cK)$.  Let $V\in\B(\cH,\cK)$ be such that
$\cK=\overline{\mathrm{span}}(\mathcal{B}V\cH)$.
For every $A\in(V^* \mathcal{B}V)'$, there exists a unique $A_1\in\mathcal{B}'$ such that
$VA=A_1V$. Furthermore, the map $\pi':A\in(V^* \mathcal{B}V)'\ni A\mapsto
A_1\in\mathcal{B}'\cap\{VV^*\}'$ is a normal surjective *-homomorphism.
\end{theorem}

Let $\cX$ and $\cY$ be C*-algebras.
We denote by $\mathrm{CP}(\cX,\cY)$ the set of completely positive linear maps
on $\cX$ to $\cY$.
Let $\cM$ be a von Neumann algebra on a Hilbert space $\cH$.
For every $T\in\mathrm{CP}(\cX,\cM)$,
we denote by $(\pi_T,\cK_T,V_T)$ the minimal Stinespring representation of $T$.
The following theorem is known as the Arveson extension theorem.

\begin{theorem}[Arveson \protect{\cite[Theorem 1.2.3]{Arv69},
Ref.~\onlinecite[Theorem 7.5]{Pau02}}]\label{ArExTh}
Let $\cX$ and $\cY$ be C*-algebras such that $\cX\subset\cY$.
Let $\cH$ be a Hilbert space.
For every $T\in\mathrm{CP}(\cX,\B(\cH))$,
there exists $\tT\in\mathrm{CP}(\cY,\B(\cH))$ such that
$\tT(X)=T(X)$, $X\in\cX$.
\end{theorem}

Let $\cM$ be a von Neumann algebra on a Hilbert space  $\cH$.
Denote by $\cM_{*}$ the predual of $\cM$, i.e., 
the space of ultraweakly continuous linear functionals on $\cM$.
Denote by $\braket{\cdot,\cdot}$ the duality pairing between $\cM_{*}$ 
(or $\cM^{*}$) and $\cM$.
We adopt the following notations:
\begin{align}
\cM_{\ast,+} &=\{ \varphi\in\cM_{\ast}| \varphi\geq 0\},\nonumber\\
\SM &= \{ \varphi\in\cM_{\ast,+}| \varphi(1)=1\}.
\end{align}

A measurable space is a pair $(S,\cF)$ of a set $S$ and a $\si$-algebra, (equivalently,
a $\si$-field, or a tribe) of subsets of $S$.  
As in some of our previous works \cite{84QC,85CA,85CC} 
a measurable space $(S,\cF)$ is also called a 
Borel space whether  $S$ is a topological space and $\cF$ is the $\si$-algebra 
generated by open subsets of $S$ or not \cite{Mac57,Tak79}.

Let $(S,\cF,\mu)$ be a finite measure space, i.e, a measurable space $(S,\cF)$ 
with a finite measure $\mu$ on $\cF$. 
Denote by $\cL(S,\mu)$ be the *-algebra of complex-valued $\mu$-measurable functions on $(S,\cF,\mu)$.
A $\mu$-measurable function $f$ is called negligible if $f(s)=0$ for $\mu$-\ae $s\in S$.
Denote by $\cN(S,\mu)$ the ideal of $\mu$-negligible functions on $S$.
Denote by $L(S,\mu)$ the quotient *-algebra modulo the negligible functions, i.e., 
$L(S,\mu)=\cL(S,\mu)/\cN(S,\mu)$.
We write $[f]=f+\cN(S,\mu)$ for any $f\in\cL(S,\mu)$.
Denote by $\cL^{1}(S,\mu)$ the space of complex-valued $\mu$-integrable 
functions on $S$.  The quotient space of $\cL^{1}(S,\mu)$
modulo the negligible functions, denoted by $L^{1}(S,\mu)$, 
is a Banach space with the $L^1$ norm defined by
$\|[f]\|_{1}=\int_S |f(s)| d\mu(s)$ for all $f\in \cL^{1}(S,\mu)$.
Denote by $M^{\infty}(S,\mu)$ the *-subalgebra of bounded complex-valued 
$\mu$-measurable functions on $S$.
A function $g\in M^{\infty}(S,\mu)$ is called $\mu$-negligible if
$g(s)=0$ for $\mu$-\ae $s\in S$.
The quotient algebra of $M^{\infty}(S,\mu)$ modulo the negligible functions,
denoted by $L^{\infty}(S,\mu)$, is a commutative W$^*$-algebra, with the predual
$L^{1}(S,\mu)$,
with respect to the essential supremum norm defined by 
$\|[g]\|_{\infty}={\mbox{ess sup}_{s\in S}}|g(s)|$ for all $g\in M^{\infty}(S,\mu)$.

\begin{definition}[\protect{CP measure space, \Cite{OOS15}, Definition 5.1}]
A triplet $(S,\cF,\Phi)$ is called a CP measure space
if it satisfies the following two conditions.

(i)  $(S,\cF)$ is a measurable space.

(ii) $\Phi$ is a $\mathrm{CP}(\cX,\cM)$-valued map on $\cF$ satisfying
\begin{equation}
\braket{\rh,\Phi(\cup_i \De_i) X}=\sum_i\braket{\rh,\Phi(\De_i) X}
\end{equation}
for any mutually disjoint sequence $\{\De_i\}_{i\in\mathbb{N}}$ in $\cF$,
$\rh\in\cM_\ast$, and $X\in\cX$.

A CP measure space $(S,\cF,\Phi)$ is called a CP measure space with barycenter
$T\in\mathrm{CP}(\cX,\cM)$ or a CP measure space of $T$ if $T=\Phi(S)$.
\end{definition}

For a normal positive linear functional $\rh$ on $\cM$,
the positive finite measure $\rh\circ\Phi$ on $S$
is defined by $(\rh\circ\Phi)(\De)=\braket{\rh,\Phi(\De)1}$ for all $\De\in\cF$.
If  $\rh$ is faithful, $L^{\infty}(S,\rh\circ\Phi)$ 
is identical with the space $L^{\infty}(S,\Phi)$ of essentially bounded 
$\Phi$-measurable functions.

\begin{lemma}[\protect{Ref.~\onlinecite[Lemma 5.3]{OOS15}}]\label{tomita5}
If $(S,\cF,\Phi)$ is a CP measure space of $T\in\mathrm{CP}(\cX,\cM)$,
then there exists a unique positive contractive linear map
$L^{\infty}(S,\Phi)\ni f\mapsto \kappa_\Phi(f)\in\pi_T(\cX)'$
satisfying
\begin{equation}
V_T^* \kappa_\Phi(f)\pi_T(X)V_T=\int f(s) d\Phi(s)X
=:\Phi(f)X,\hspace{5mm}f\in L^{\infty}(S,\Phi),X\in\cX,
\end{equation}
i.e., for every $\rh\in\cM_\ast$,
\begin{equation}
\bracket{\rh,V_T^* \kappa_\Phi(f)\pi_T(X)V_T}
=\int f(s) d\bracket{\rh,\Phi(s)X},
\hspace{5mm}f\in L^{\infty}(S,\Phi),X\in\cX. \label{tomita0}
\end{equation}
Furthermore, if $f\in L^{\infty}(S,\Phi)_+$ satisfies $\kappa_\Phi(f)=0$, then $f=0$.
If $L^{\infty}(S,\Phi)$ is equipped with the $\si(L^{\infty}(S,\mu),L^1(S,\mu))$-topology
and $\pi_T(\cX)'$ with the weak topology, 
where $\mu=\varphi\circ\Phi$ for some normal faithful state $\varphi$ on $\cM$,
then the map $\kappa_\Phi$ 
is continuous.
\end{lemma}

\begin{theorem}[\protect{Ref.~\onlinecite[Part I, Chapter 4, Theorem 3]{Dix81};
Ref.~\onlinecite[Chapter IV, Theorem 5.5]{Tak79}}] \label{vNhom}
Let $\cM_1$ and $\cM_2$ be von Neumann algebras on Hilbert spaces $\cH_1$
and $\cH_2$, respectively. If $\pi$ is a normal *-homomorphism of $\cM_1$ onto
$\cM_2$, there exist a Hilbert space $\cK$, a projection $E$ of $\cM_1'\otimes
\B(\cK)$, and an isometry $U$ of $\cH_2$ onto $E(\cH_1\otimes\cK)$
 such that
\begin{equation}
\pi(M)=U^*j_E(M\otimes 1_\cK)U,\hspace{5mm}M\in\cM_1,
\end{equation}
where $j_E$ is a CP map of $\B(\cH_1\otimes\cK)$ onto
$E\B(\cH_1\otimes\cK)E$ defined by
$j_E(X)=EXE$, $X\in\B(\cH_1\otimes\cK)$.
\end{theorem}
We also use the following form of Theorem \ref{vNhom}:
\begin{corollary}[] \label{vNhom2}
Let $\cH_1$ and $\cH_2$ be Hilbert spaces.
If $\pi$ is a normal representation of $\B(\cH_1)$ on
$\cH_2$, there exist a Hilbert space $\cK$
and a unitary operator $U$ of $\cH_2$ onto $\cH_1\otimes\cK$ such that
\begin{equation}
\pi(X)=U^*(X\otimes 1_\cK)U,\hspace{5mm}X\in\B(\cH_1).
\end{equation}
\end{corollary} 

\section{Completely Positive Instruments and Quantum Measuring Processes}
\label{se:completely}

Let $\cM$ be a von Neumann algebra on a Hilbert space $\cH$
and $(S,\cF)$ a measurable space.
In the rest of this paper, we assume that von Neumann algebras are $\si$-finite.
We denote by $P(\cM_{*})$ [or, $\CP(\cM_{*})$] the set of positive [or, completely positive \cite{Tak79}] 
linear maps on $\cM$
and $P_n(\cM)$ [or, $\CP_n(\cM)$] the set of normal positive [or, completely positive, resp.] 
linear maps on $\cM$.
Note that every $\Ph\in P_n(\cM)$ has the unique predual map $\Ph_{*}\in P(\cM_*)$
such that $(\Ph_{*})^{*}=\Ph$ and the corresponding $\Ph\mapsto\Ph_{*}$ is a bijection
between $P_n(\cM)$ and $P(\cM_{*})$ and also implements a bijection between
${\rm CP}_n(\cM)$ and ${\rm CP}(\cM_{*})$.
Now we introduce the concept of instrument, which plays a central role in quantum
measurement theory.

\begin{definition}[Instruments, Davies-Lewis \protect{\cite[Section 3]{DL70}}]
An instrument $\cI$ for $(\cM,S)$ is a $P(\cM_{*})$-valued map on $\cF$
satisfying the following two conditions.

(i) $\|\cI(S)\rh\|=\|\rh\|$ for all $\rh\in\cM_{*}$.

(ii) For each countable mutually disjoint sequence $\{\De_j\}\subset\cF$,
\begin{equation}
\cI(\cup_j \De_j)\rh=\sum_j \cI(\De_j)\rh
\end{equation}
for all $\rh\in\cM_{*}$.  

An instrument $\cI$ for $(\cM,S)$ is called completely
positive (CP) if $\cI(\De)$ is a completely positive map 
on $\cM_{*}$ for every $\De\in\cF$.
We denote by $\mathrm{CPInst}(\cM,S)$ the set of CP instruments for $(\cM,S)$.
\end{definition}

Let $\cI^{*}(\De)$ be the dual map on $\cM$ of  
$\cI(\De)$ defined by $\braket{\rh,\cI^{*}(\De)M}=
\braket{\cI(\De)\rh,M)}$ for all $\rh\in\cM_{*}$ and $M\in\cM$.
In this case, $\cI^{*}$ is a $P_n(\cM)$-valued measure on $\cF$.
We also write $\cI(M,\De)=\cI^{*}(\De)M$
for all $\De\in\cF$ and $M\in\cM$.
Then, a map $\cI(\cdot,\cdot)$ of $\cM\times\cF$ into $\cM$ arises from
an instrument in this way if and only if the following three conditions hold \cite{Dav70}.

(i) $M\mapsto \cI(M,\De)$ is a normal positive linear map on $\cM$ for all $\De\in\cF$.

(ii) $\De\mapsto\bracket{\rh,\cI(M,\De)}$ is a countably additive finite signed measure
for all $\rh\in\cM_{*}$ and $M\in\cM$.

(iii) $\cI(1,S)=1$.

If $\cI$ is completely positive, $(S,\cF,\cI^{*})$ is a CP measure space.
For any normal state $\rh$ on $\cM$, denote 
by $\cI\rh$ the $\cM_{*}$-valued measure on $(S,\cF)$ defined by 
$\cI\rh(\De)=\cI(\De)\rh$, where $\De\in\cF$, and 
by $\|\cI\rh\|$ the probability measure on $(S,\cF)$ defined by 
$\|\cI\rh\|(\De)=\|\cI(\De)\rh\|$.
For any $M\in\cM$, $\bracket{\cI\rh,M}$ stands for the signed measure
such that $\bracket{\cI\rh,M}(\De)=\bracket{\rh,\cI(M,\De)}$.
We have $\|\cI\rh\|(\De)=\bracket{\rh,\cI(1,\De)}$.

To discuss the role of CP instruments in quantum measurement theory, here
we assume that the system $\bS$ of interest is described by a von Neumann algebra 
$\cM$ on a Hilbert space $\cH$; observables of $\bS$ are represented by
self-adjoint operators affiliated with $\cM$ and states of $\bS$ are described 
by normal states on $\cM$. 
Consider an apparatus $\bA(\bx)$ measuring $\bS$ having the output variable $\bx$ 
with values in a measurable space $(S,\cF)$.
In standard experimental situations, the measuring apparatus $\bA(\bx)$
is naturally assumed to have the following statistical properties: 
(i) The probability $\mathrm{Pr}\{\bx\in\De\|\rh\}$ of the outcome event $\bx\in\De $
for any input state $\rh$ of $\bS$.
(ii) The state change $\rh\mapsto\rh_{\{\bx\in\De\}}$ from any input state $\rh$
to the output state $\rh_{\{\bx\in\De\}}$ given the outcome event $\bx\in\De$
provided that $\mathrm{Pr}\{\bx\in\De\|\rh\}>0$, otherwise 
$\rh_{\{\bx\in\De\}}$ is indefinite.
Consider the successive measurements carried out by two apparatuses $\bA(\bx)$ with the output 
variable $\bx$ and $\bA({\bf y})$ with the output variable $\by$ in this order,
where ${\bf y}$ values in a measurable space $(S',\cF')$.
Then, the joint probability distribution $\mathrm{Pr}\{({\bf y},\bx)\in\De\|\rh\}$
of $\bx$ and $\by$ on $(S'\times S,\cF'\times \cF)$ is uniquely 
determined by the formula
\begin{equation}\label{JPD}
\Pr\{\by\in\De_2,\bx\in\De_1\|\rh\}
=\Pr\{\by\in\De_2\|\rh_{\{\bx\in\De_1\}}\}
\Pr\{\bx\in\De_1\|\rh\}
\end{equation}
for all $\De_1\in\cF$ and $\De_2\in\cF'$,  where 
$\Pr\{\by\in\De_2,\bx\in\De_1\|\rh\}=
  \Pr\{(\by,\bx)\in(\De_2,\De_1)\|\rh\}$.
We naturally assume that
the joint probability distribution $\mathrm{Pr}\{{\bf y}\in\De_2,\bx\in\De_1\|\rh\}$ 
is an affine function of $\rh\in\SM$.
For every $\De\in\cF$,
we then define a map $\cI(\De):\SM\to\cM_{\ast,+}$ by
\beq
\cI(\De)\rh=\mathrm{Pr}\{\bx\in\De\|\rh\}\rh_{\{\bx\in\De\}}
\eeq
for all $\rh\in\SM$.
Under the above assumptions, it is shown in 
\Cites{97OQ,04URN} that $\rh\mapsto\cI(\De)\rh$
is an affine map of $\SM$
for all $\De\in\cF$, so that it uniquely extends to a positive
linear map on $\cM_\ast$ satisfying $\bracket{\cI(S)\rh,1}=\bracket{\rh,1}$
for all $\rh\in\cM_\ast$, and
$ 
\bracket{\cI(\cup_j\De_j)\rh,M}=\sum_j \bracket{\cI(\De_j)\rh,M}
$ 
for all $M\in\cM$, $\rh\in\cM_{\ast}$, and mutually disjoint sequence $\{\De_j\}\subset\cF$.
Then, the map $\De\to\cI(\De)$ is an instrument for $(\cM,S)$.
Thus, every measuring apparatus $\bA(\bx)$ defines an instrument $\cI$
satisfying the following characteristic conditions proposed by Davies and Lewis
\cite{DL70} (see \Cite{04URN}, Sections 2.2--2.6, for more detailed discussions).

(i) $\Pr\{\bx\in\De\|\rh\}=\|\cI(\De)\rh\|$.

(ii) $\rh_{\{\bx\in\De\}}=\dfrac{\cI(\De)\rh}{\|\cI(\De)\rh\|}$.

For quantum systems with finite degrees of freedom,  
we can further advance our analysis of statistical properties of measuring
apparatuses.
In this case, every observable $A$ of $\bS$ can be identified with the observable 
$A\otimes I$ of the extended system $\bS+\bS'$ with any external system $\bS'$.
By the same token, it is natural to require the trivial extendability condition stating
that every instrument $\cI$ for $\bS$ can be extended to an instrument $\cI'$ for 
$\bS+\bS'$ such that $\cI'(\De)=\cI(\De)\otimes \id$ for all $\De\in\cF$.
Then, it is concluded that the instrument $\cI$ should be completely positive, if it
describes a physically realizable measurement at all.
See \Cite{04URN}, Section 2.9 for more detailed discussion.
We shall reconstruct the above argument for algebraic quantum field theory
in the last section.

The next problem is to determine which CP instrument arises from a measuring 
apparatus. To discuss this problem, we introduce the concept of measuring process 
as a general class of models of quantum measurement for the system described 
by a von Neumann algebra.

Let $\cM$ and $\cN$ be von Neumann algebras.
For every $\si\in \cN_\ast$, the normal unital CP map
$\id\otimes\si:\cM \ootimes  \cN
\to\cM$ is defined by
$\braket{\rh,(\id\otimes\si) X}=\braket{\rh\otimes\si,X}$
for all $X\in\cM \ootimes  \cN$
and $\rh\in\cM_\ast$.

\begin{definition}[Measuring processes,  \Cite{84QC}, Definition 3.1]
A measuring process $\M$ for $(\cM,S)$
is a 4-tuple $\M=(\cK,\si,U,E)$
consisting of a Hilbert space $\cK$, a normal
state $\si$ on $\B(\cK)$,
a unitary operator $U$ on $\cH\otimes\cK$,
and a spectral measure $E:\cF\to \B(\cK)$
satisfying 
\beql{def-instrument}
(\id\otimes\si)[U^*(M\otimes E(\De))U]\in\cM
\eeq
for every $M\in\cM$ and $\De\in\cF$.
\end{definition}

As shown in Section 5 in Ref.~\onlinecite{84QC} for any measuring 
process $\M=(\cK,\si,U,E)$ the relation 
\beql{realization}
\cI_{\bbM}(X,\De)=(\id\otimes\si)[U^*(X\otimes E(\De))U],
\eeq
where $X\in \B(\cH)$ and $\De\in\cF$, defines a CP instrument $\cI_{\bbM}$ for
$(\B(\cH),S)$, which describes the statistical properties of the measuring process
$\M$. Condition \eq{def-instrument} ensures that the restriction $\cI_{\bbM}|_{\cM}$
of $\cI_{\bbM}$ to $\cM$ defined by $\cI_{\bbM}|_{\cM}(M,\De)=\cI_{\bbM}(M,\De)$ for all 
$\De\in\cF$ and $M\in\cM$ is a CP instrument for $(\cM,S)$.  
We say that a measuring process  $\M=(\cK,\si,U,E)$ for $(\cM,S)$ realizes 
an instrument $\cI$ for $(\cM,S)$ if $\cI=\cI_{\bbM}|_{\cM}$.

Now, the converse problem is posed naturally: Does every CP instrument on $\B(\cH)$
arise from a measuring apparatus $\bA(\bx)$ for $\bS$?
In the previous investigation \cite{84QC} this problem was solved affirmatively
as follows.

\begin{theorem}[\protect{\Cite{84QC}, Theorem 5.1}]\label{th:realization}
Let $\cH$ be a  Hilbert space and $(S,\cF)$ be a measurable space.
For every CP instrument $\cI$ for $(\cB(\cH),S)$ there exists a measuring process 
$\M=(\cK,\si,U,E)$ for $(\cB(\cH),S)$ that realizes $\cI$.
\end{theorem}

Thus, the measurement described by any CP instrument $\cI$ is 
realized by an interaction described by a unitary operator $U$  with the probe 
prepared in a state $\si$  and the subsequent measurement 
of the meter observable described by the spectral measure $E$  in the probe,
and we can conclude that the description of measurement by every CP instrument is 
consistent with the description of measurement by the unitary evolution of the system 
plus the probe based on von Neumann's postulates for quantum mechanics.
We refer the reader to \Cites{04URN,14MFQ} for detailed expositions on quantum 
measurement theory for systems with finite degrees of freedom.
We now try to generalize the above correspondence between CP instruments and
measuring processes to quantum systems of infinite degrees of freedom as follows.

We first observe that every CP instrument admits the following representation.

\begin{proposition}[\protect{\Cite{84QC}, Proposition 4.2}]\label{Naimark-Ozawa}
For any CP instrument $\cI$ for $(\cM,S)$, 
there are a Hilbert space $\cK$,
a nondegenerate normal faithful representation
$E:L^{\infty}(S,\cI)\to \B(\cK)$,
a nondegenerate normal representation
$\pi:\cM\to\B(\cK)$
and an isometry $V\in\B(\cH,\cK)$ satisfying
\begin{align}
\cI(M,\De) &=V^* E([\chi_{\De}])\pi(M)V, \label{CPrep}\\
E([\chi_{\De}])\pi(M) &=\pi(M)E([\chi_{\De}])
\label{CPrep2}
\end{align}
for any $\De\in\cF$ and $M\in\cM$.
\end{proposition}

\begin{proof}
By Lemma \ref{tomita5}, there exists a
positive contractive linear map $\ka: L^{\infty}(S,\cI)\to \pi_T(\cX)'$ such that
\begin{equation}
\cI(M,\De)=\cI^*(\De)M=V_T^* \ka([\ch_{\De}])\pi_T(M)V_T
\end{equation}
for all $M\in\cM$ and $\De\in\cF$.
Let $(E,\cK,W)$ be the minimal Stinespring representation of $\ka$.
By Theorem \ref{CommLift}, there exists a nondegenerate normal
representation $\pi$ of $\cM$ on $\cK$ such that
\begin{align}
 \pi(M)W&=W\pi_T(M), \\
E([\ch_{\De}])\pi(M) &=\pi(M)E([\ch_{\De}]) 
\end{align}
for all $M\in\cM$ and $\De\in\cF$. 
We denote $WV_T$ by $V$,
which is seen to be an isometry.  Then, we have
\begin{align}
\cI(M,\De)&=V_T^* \ka([\ch_{\De}])\pi_T(M)V_T
=V_T^* W^* E([\ch_{\De}])W\pi_T(M)V_T\nonumber\\
 &= (WV_T)^* E([\ch_{\De}])\pi(M)WV_T=V^* E([\ch_{\De}])\pi(M)V
\end{align}
for all $M\in\cM$ and $\De\in\cF$,
which concludes the proof.
\end{proof}

\begin{remark}
An alternative proof using Proposition 4.2 of \Cite{84QC}
instead of Lemma \ref{tomita5} runs as follows.
From Proposition 4.2 of \Cite{84QC} we can construct a Hilbert space $\cK$,
a spectral measure $E_0: \cF \to \B(\cK)$, 
a nondegenerate normal representation 
$\pi:\cM\to\B(\cK)$,
and an isometry $V\in\B(\cH,\cK)$ 
satisfying relations analogous to (\ref{CPrep}) and (\ref{CPrep2}).  
By modifying the construction it is easy to see that we can assume 
that $\cK$ is spanned by 
$\{E_0(\De)\pi(M)V\xi\mid \De\in\cF, M\in\cM,\xi\in\cH\}$.
Then, $E_0(\De)=0$ if and only if $\cI(\De)=0$ for any $\De\in\cF$,
and the relation $E([f])=\int_{S} f(s) dE_0(s)$ defines a nondegenerate normal faithful representation $E:L^{\infty}(S,\cI)\to \B(\cK)$,
which satisfies all the assertions of Proposition  \ref{Naimark-Ozawa}.
\end{remark}

Let $\cM$ be a von Neumann algebra and $(S,\cF)$ a measurable space.
Let $\cM\aotimes L^{\infty}(S,\cI)$ be the algebraic tensor product of $\cM$
and $L^{\infty}(S,\cI)$.   Any $\rh\in\cM_{*}$ and $M\in\cM$
defines a finite signed measure $\De\mapsto\bracket{\rh,\cI(M,\De)}$ on $(S,\cF)$
absolutely continuous with respect to $\cI$.
Thus, the relation
\beq
\Phi_\cI(M\otimes [f])=\int_{S}f(s)\, d\bracket{\rh,\cI(M,s)},
\eeq
where $M\in\cM$ and $f\in M^{\infty}(S,\cI)$,
defines a unique positive linear map $\Phi_{\cI}$ of 
$\cM\aotimes L^{\infty}(S,\cI)$
into $\cM$.  The positive map $\Phi_{\cI}$ determined above is called the linear extension 
of $\cI$.  The following proposition shows that if $\cI$ is a CP instrument, $\Phi_{\cI}$ can 
be further extended to the unique C*-norm closure of 
$\cM\aotimes L^{\infty}(S,\cI)$.

\begin{proposition}\label{th:linear}
Let $\cM$ be a von Neumann algebra and $(S,\cF)$ a measurable space.
Every CP instrument $\cI$ for $(\cM,S)$ can be uniquely extended to a completely 
positive map $\Psi_{\cI}$ of $\cM\motimes  L^{\infty}(S,\cI)$ into $\cM$ such that 
$\Psi_{\cI}(M \otimes [\ch_{\De}])=\cI(M,\De)$ for all $\De\in\cF$ and $M\in\cM$.
In this case, $\Psi_{\cI}$ is binormal,
i.e., normal on each tensor factor, and extends $\Phi_{\cI}$.
\end{proposition}

\begin{proof}
Let $\cI$ be a CP instrument for $(\cM,S)$.
By Proposition \ref{Naimark-Ozawa} and 
by a universal property of binormal tensor product (Proposition \ref{Max}),
there exists a representation $\widetilde{\pi}:\cM\botimes L^{\infty}(S,\cI)
\to\B(\cK)$
such that $\widetilde{\pi}(M\otimes f)=\pi(M)E(f)$ for every $M\in\cM$ and $f\in L^{\infty}(S,\cI)$.
We can define a CP map $\Psi_\cI:
\cM\botimes  L^{\infty}(S,\cI)\to\B(\cH)$ by
\begin{equation}
\Psi_\cI(X)=V^* \widetilde{\pi}(X)V,\hspace{5mm}
X\in \cM\botimes  L^{\infty}(S,\cI),
\end{equation}
which is binormal. Since every commutative C*-algebra is nuclear, it holds that
\begin{equation}
\cM\botimes L^{\infty}(S,\cI)
=\cM\motimes L^{\infty}(S,\cI)
 (=\cM\otimes_{\max}L^{\infty}(S,\cI)),
\end{equation}
and the assertion follows.
\end{proof}

\sloppy
Note that $\cM\motimes L^{\infty}(S,\cI)$ is a weakly 
dense C*-subalgebra of the von Neumann algebra $\cM \ootimes L^{\infty}(S,\cI)$.
By the Arveson extension theorem (Theorem \ref{ArExTh}), 
$\Psi_\cI$ can be extended to a (not necessarily normal) CP map $\widetilde{\Psi_\cI}:
\cM \ootimes   L^{\infty}(S,\cI)\to\B(\cH)$, which satisfies
$\widetilde{\Psi_\cI}|_{\cM\motimes L^{\infty}(S,\cI)}=\Psi_\cI$.
In the following, we shall show that a CP instrument $\cI$ has a measuring process
$\bbM$  if and only if the binormal CP map ${\Psi_\cI}$ on $\cM\motimes L^{\infty}(S,\cI)$
has a normal extension to $\cM \ootimes   L^{\infty}(S,\cI)$.
Motivated as above we introduce the following definitions, where we call the above property
the normal extension property.

\begin{definition}[Normal extension property]
Let $\cI$ be a CP instrument for $(\cM,S)$.

(i) $\cI$ has the normal extension property (NEP) if there exists a unital normal CP map
$\widetilde{\Psi_\cI}:\cM \ootimes   L^{\infty}(S,\cI)\to
\B(\cH)$ such that
$\widetilde{\Psi_\cI}|_{\cM\motimes  L^{\infty}(S,\cI)}=\Psi_\cI$.

(ii) $\cI$ has the unique normal extension property (UNEP) if
there exists a unique unital normal CP map
$\widetilde{\Psi_\cI}:\cM \ootimes   L^{\infty}(S,\cI)\to
\B(\cH)$ such that
$\widetilde{\Psi_\cI}|_{\cM\motimes  L^{\infty}(S,\cI)}=\Psi_\cI$.
\end{definition}

We denote by
$\mathrm{CPInst}_{\mathrm{NE}}(\cM,S)$ the set of CP instruments for $(\cM,S)$ with  the NEP.
We gave the name  ``normal extension property" in the light of the
unique extension property \cite{Arv08} used in operator system theory.
Let $\cI$ be a CP instrument for $(\cM,S)$ with  the NEP and 
$\widetilde{\Psi_\cI}:\cM \ootimes   L^{\infty}(S,\cI)\to
\B(\cH)$
a unital normal CP map such that
$\widetilde{\Psi_\cI}|_{\cM\motimes L^{\infty}(S,\cI)}=\Psi_\cI$.
Since $\cM \ootimes   L^{\infty}(S,\cI)=
\overline{\cM\otimes_{\alg} L^{\infty}(S,\cI)}^{uw}$ and
$\widetilde{\Psi_\cI}$
is ultraweakly continuous,
it follows that
$\widetilde{\Psi_\cI}(\cM \ootimes   L^{\infty}(S,\cI))\subset\cM$.
Thus, any normal extension $\widetilde{\Psi_\cI}$ of $\Psi_\cI$ ranges in $\cM$.

To show that every CP instrument $\cI$ having a measuring process $\bbM$
has the NEP, we begin with examining faithful measuring processes defined below.

\begin{definition}[Faithfulness of measuring process]
A measuring process $\M=(\cK,\si,U,E)$ for $(\cM,S)$
is said to be faithful if there exists a normal faithful representation
$\widetilde{E}:L^{\infty}(S,\cI_\M)\to \B(\cK)$
such that $\widetilde{E}([\chi_\De])=E(\De)$ for all $\De\in\cF$.
\end{definition}

Let  $\M=(\cK,\si,U,E)$ be a faithful measuring process for 
$(\cM,S)$ and $\cI_{\bbM}$ the CP instrument of $\bbM$.
Then, a unital normal CP map
$\Psi_\M:\cM \ootimes  
L^{\infty}(S,\cI_\M)\to \cM$ is defined by
\begin{equation}
\Psi_\M(X)=(\id\otimes\si)[U^*((\id_{\cM}\otimes \widetilde{E})(X))U]
\end{equation}
for all $X\in\cM \ootimes L^{\infty}(S,\cI_\M)$,
where $\widetilde{E}$ is a normal faithful representation of
$L^{\infty}(S,\cI_\M)$ on $\cK$
such that $\widetilde{E}([\chi_\De])=E(\De)$ for all $\De\in\cF$.
In this case, we have 
\beq
\Psi_\M(M\otimes[\chi_{\De}])=
\cI_{\bbM}(M,\De)=\Psi_{\cI_\M}(M\otimes[\chi_{\De}])
\eeq
for all $M\in\cM$ and $\De\in\cF$, and hence 
$\Psi_\M|_{\cM\motimes L^{\infty}(S,\cI)}=\Psi_\mathcal{I_{\bbM}}$.
Thus, the CP instrument $\cI_{\bbM}$ has the NEP.

Now we shall show that an instrument has a measuring process if and only if
it is completely positive and has the NEP.

\begin{theorem}\label{NEP}
For a CP instrument $\cI$ for $(\cM,S)$, the following conditions are equivalent:

(i) $\cI$ has  the NEP.

(ii) $\cI$ has the UNEP.

(iii) There exists a CP instrument $\widetilde{\cI}$
for $(\B(\cH),S)$ such that $\widetilde{\cI}(M,\De)=\cI(M,\De)$ for all $\De\in\cF$
and $M\in\cM$.

(iv) There exists a faithful measuring process $\M=(\cK,\si,U,E)$
for $(\cM,S)$ such that
\begin{equation}
\cI(M,\De)=(\id\otimes\si)[U^*(M\otimes E(\De))U ]
\end{equation}
for all $\De\in\cF$ and $M\in\cM$.

(v) There exists a measuring process $\M=(\cK,\si,U,E)$
for $(\cM,S)$ such that
\begin{equation}
\cI(M,\De)=(\id\otimes\si)[U^*(M\otimes E(\De))U ]
\end{equation}
for all $\De\in\cF$ and $M\in\cM$.
\end{theorem}

\begin{proof} 
(ii)$\Rightarrow$(i) Obvious.

(i)$\Rightarrow$(ii) Let $T_1,T_2:\cM \ootimes   L^{\infty}(S,\cI)\to
\B(\cH)$ be normal CP maps such that
\begin{equation}
T_j|_{\cM\motimes L^{\infty}(S,\cI)}=\Psi_\cI
\end{equation}
for $j=1,2$. By assumption, for every $M\in\cM\otimes_{\alg}L^{\infty}(S,\cI)$, we have
$T_1(M)=T_2(M)$. Since $\cM\otimes_{\alg}L^{\infty}(S,\cI)$ is dense in
$\cM \ootimes   L^{\infty}(S,\cI)$, and $T_1$ and $T_2$ are normal
on $\cM \ootimes   L^{\infty}(S,\cI)$,
it is seen that $T_1$ is equal to $T_2$ on $\cM \ootimes   L^{\infty}(S,\cI)$.

(iv)$\Rightarrow$(v) Obvious.

(v)$\Rightarrow$(iii) Obvious.

(iv)$\Rightarrow$(i) Let $\M=(\cK,\si,U,E)$ be a faithful measuring process
for $(\cM,S)$ such that
\begin{equation}
\cI(M,\De)=(\id\otimes\si)[U^*(M\otimes E(\De))U]=\Psi_{\M}(M\otimes[\chi_\De])
\end{equation}
for all $\De\in\cF$ and $M\in\cM$.
Then $\Psi_{\M}$ satisfies
$\Psi_{\M}|_{\cM\motimes L^{\infty}(S,\cI)}=\Psi_\cI$.

(i)$\Rightarrow$(iii) By assumption, there exists a unital normal CP map
$\widetilde{\Psi_\cI}:\cM \ootimes   L^{\infty}(S,\cI)\to\cM$
such that
$\widetilde{\Psi_\cI}|_{\cM\motimes L^{\infty}(S,\cI)}=\Psi_\cI$.
There then
exists a minimal Stinespring representation $(\pi,\cL_1,W_1)$ of $\widetilde{\Psi_\cI}$,
i.e.,
\begin{equation}
\widetilde{\Psi_\cI}(X)=W_1^*\pi(X)W_1,\hspace{5mm}
X\in\cM \ootimes   L^{\infty}(S,\cI).
\end{equation}
Furthermore, by Theorem \ref{vNhom}, there exist a Hilbert space $\cL_2$,
a projection $E$ of $(\cM \ootimes   L^{\infty}(S,\cI))'$
$\otimes\B(\cL_2)$
and an isometry $W_2:\cL_1\to E(\cH\otimes L^2(S,\cI)\otimes \cL_2)$
such that
\begin{equation}
\pi(X)=W_2^*j_E(X\otimes 1_{\cL_2})W_2,\hspace{5mm}
X\in\cM \ootimes   L^{\infty}(S,\cI),
\end{equation}
where a normal CP map
$j_E:\B(\cH\otimes L^2(S,\cI)\otimes \cL_2)\to
E(\B(\cH\otimes L^2(S,\cI)\otimes \cL_2))E$ is defined by
$j_E(X)=EXE$, $X\in\B(\cH\otimes L^2(S,\cI)\otimes \cL_2)$.
We then define a normal CP map $\widetilde{\pi}:
\B(\cH) \ootimes   L^{\infty}(S,\cI)
\to \B(\cL_1)$ by
\begin{equation}
\widetilde{\pi}(X)=W_2^*j_E(X\otimes 1_{\cL_2})W_2,\hspace{5mm}
X\in\B(\cH) \ootimes   L^{\infty}(S,\cI).
\end{equation}
A CP instrument $\widetilde{\cI}$ for $(\B(\cH),S)$ is defined by
\begin{equation}
\widetilde{\cI}(X,\De)=W_1^*\widetilde{\pi}(X\otimes[\chi_\De])W_1
\end{equation}
for every $X\in\B(\cH)$ and $\De\in\cF$.
For every $M\in\cM$ and $\De\in\cF$, it is seen that
\begin{align}
\widetilde{\cI}(M,\De) &= W_1^*\widetilde{\pi}(M\otimes[\chi_\De])W_1\nonumber\\
 &= W_1^*\pi(M\otimes[\chi_\De])W_1=\widetilde{\Psi_\cI}(M\otimes[\chi_\De])=\cI(M,\De)
\end{align}
for every $M\in\cM$ and $\De\in\cF$.

(iii)$\Rightarrow$(iv) Let $\widetilde{\cI}$ be a CP instrument 
for $(\B(\cH),S)$ such that $\widetilde{\cI}(M,\De)=\cI(M,\De)$ for all $\De\in\cF$ and $M\in\cM$.
Let $\tT=\widetilde{\cI}^*(S)$.
By Corollary \ref{vNhom2}, a normal representation of $\B(\cH)$
is unitarily equivalent to the representation $\id\otimes 1_{\cL}$,
where $\cL$ is a Hilbert space. 
Therefore, there exist a Hilbert space $\cL_1$ and
a unitary operator $W_1:\cK_{\tT}\to \cH\otimes\cL_1$ such that
\begin{equation}
\tT(X)=V_{\tT}^* W_1^*(X\otimes 1)W_1V_{\tT}
\end{equation}
for all $X\in \B(\cH)$.
By Lemma \ref{tomita5}, there exists a positive contractive map
$\ka: L^{\infty}(S,\cI)\to
\B(\cL_1)$ such that
\begin{equation}
\widetilde{\cI}(X,\De)=V_{\tT}^* W_1^*
(X\otimes \ka([\ch_{\De}]))W_1V_{\tT}
\end{equation}
for all $X\in\B(\cH)$ and $\De\in \cF$,
and that, if $f\in L^{\infty}(S,\cI)_+$ satisfies $\ka(f)=0$ then $f=0$.
Let $(E_0,\cL_2,W_2)$ be the minimal Stinespring representation of
$\ka$. Then $E_0$ is a normal faithful representation of $L^{\infty}(S,\cI)$
on $\cL_2$.
Denote $W_2W_1V_{\tT}$ by $V$.
It holds that
\begin{equation}
\widetilde{\cI}(X,\De)=V^* (X\otimes E_0([\ch_{\De}]))V
\end{equation}
for all $X\in\B(\cH)$ and $\De\in\cF$.
Let $\cL_3=\cH\otimes\cL_2$.
Let  $\eta_3$ be a unit vector of $\cL_3$,
and $\eta_2$ a unit vector of $\cL_2$.
We define an isometry $U_0$ from 
$\cH\otimes\mathbb{C}\eta_2\otimes\mathbb{C}\eta_3$ to
$\cH\otimes\cL_2\otimes\cL_3$ by
\begin{equation}
U_0(\xi\otimes\eta_2\otimes\eta_3)=V\xi\otimes\eta_3,
\end{equation}
for all $\xi\in\cH$. 
Then, we have
\begin{equation}
\dim(\cH\otimes\cL_2\otimes\cL_3
\ominus\cH\otimes\mathbb{C}\eta_2\otimes\mathbb{C}\eta_3)=
\dim(\cH\otimes\cL_2\otimes\cL_3
\ominus U_0(\cH\otimes\mathbb{C}\eta_2\otimes\mathbb{C}\eta_3)),
\end{equation}
since both sides equal $\dim(\cH\otimes\cL_2\otimes\cL_3)$
if $\dim(\cL_3)$ is infinite, and otherwise they equal
$\dim(\cL_3)^2-\dim(\cH)$.
It follows that  $U_0$ can be extended to a unitary operator $U$ on $\cH\otimes\cL_2\otimes\cL_3$.
We then define a Hilbert space $\cK$ by
$\cK=\cL_2\otimes\cL_3$, a state $\si$ on $\B(\cK)$ by
\begin{equation}
\si(Y)=\langle \eta_2\otimes\eta_3|Y(\eta_2\otimes\eta_3) \rangle
\end{equation}
for all $Y\in \B(\cK)$, and a spectral measure $E:\cF\to
\B(\cK)$ by
\begin{equation}
E(\De)=E_0([\chi_\De])\otimes 1_{\cL_3}
\end{equation}
for all $\De\in\cF$. We have
\begin{align}
\langle \xi| \widetilde{\cI}(X,\De)\xi\rangle
&=\langle \xi| V^*(X\otimes E_0([\chi_\De]))V\xi\rangle \nonumber\\
 &= \langle V\xi\otimes\eta_3|(X\otimes E_0([\chi_\De]))(V\xi \otimes\eta_3)\rangle\nonumber \\
 &= \langle U(\xi\otimes\eta_2\otimes\eta_3)|(X\otimes E_0([\chi_\De])\otimes 1_{\cL_3})
 U(\xi\otimes\eta_2\otimes\eta_3)\rangle\nonumber \\
 &= \langle \xi\otimes\eta_2\otimes\eta_3|U^*(X\otimes E(\De))U
 (\xi\otimes\eta_2\otimes\eta_3)\rangle\nonumber \\
 &= \langle \xi|\{(\id\otimes\si)[U^*(X\otimes E(\De))U]\}\xi\rangle
\end{align}
for all $X\in\B(\cH)$, $\De\in\cF$, and $\xi\in\cH$.
By the definition of $E$,  there exists a
normal faithful representation $\widetilde{E}$ of $L^{\infty}(S,\cI)$ on $\cK$
such that $\widetilde{E}([\chi_\De])=E(\De)$ for all $\De\in\cF$.
Thus there exists a faithful measuring process
$\M=(\cK,\si,U,E)$ for $(\cM,S)$ such that
\begin{equation}
\widetilde{\cI}(X,\De)=(\id\otimes\si)[U^*(X\otimes E(\De))U ]
\end{equation}
for all $X\in\B(\cH)$ and $\De\in\cF$.
Therefore, for every $M\in\cM$ and $\De\in\cF$, it holds that
\begin{equation}
\cI(M,\De)=(\id\otimes\si)[U^*(M\otimes E(\De))U ].
\end{equation}
\end{proof}

\begin{remark}
From the above proof, it can be seen that every CP instrument $\cI$ for $(\cM,S)$
with the NEP has a measuring process $\M=(\cK,\si,U,E)$ 
such that $\si$ is a pure state
of $\B(\cK)$.  If $\cM$ is a von Neumann algebra on a separable Hilbert space 
$\cH$ and $(S,\cF)$ is a standard Borel space, i.e.,
a Borel space associated to a Polish space, it can be shown that every CP instrument 
for $(S,\cM)$ with the NEP has a measuring process $\M=(\cK,\si,U,E)$ such that 
$\cK$ is separable and $\si$ is a pure state on $\B(\cK)$ (cf.~\Cite{84QC},
Corollary 5.3).
\end{remark}

\begin{remark}
We can directly prove the implication (iii)$\Rightarrow$(v) without assuming the $\si$-finiteness of
$\cM$ by applying Theorem \ref{th:realization} to the CP instrument $\widetilde{\cI}$ for $(\cB(\cH),S)$ assumed in (iii) to obtain a measuring process $\bbM=(\cK,\si,U,E)$ 
for $(\cB(\cH),S)$ that realizes $\cI$.
Nevertheless, to show (v)$\Rightarrow$(iv) we still need to rework the proof of Theorem \ref{th:realization} using Proposition \ref{Naimark-Ozawa}, where $\cM$ is assumed $\si$-finite, 
similarly with the above proof of the implication (iii)$\Rightarrow$(iv).
\end{remark}

We note that Theorem \ref{NEP} is also a natural generalization
of Raginsky\cite[Theorem IV.2]{Rag03}.
The following is an immediate consequence of Theorem \ref{NEP}.

\begin{corollary}
For any instrument $\cI$ for $(\cM,S)$, the following conditions are equivalent:

(i) For every $\De\in\cF$, $\cI(\De)$ is
completely positive and $\cI$ has  the NEP.

(ii) There exists a measuring process $\M=(\cK,\si,U,E)$
for $(\cM,S)$ such that
\begin{equation}
\cI(M,\De)=(\id\otimes\si)[U^*(M\otimes E(\De))U ]
\end{equation}
for all $\De\in\cF$ and $M\in\cM$.
\end{corollary}

Many different measuring processes may describe essentially the same
measurement, where two measurements are considered to be essentially 
the same if they are indistinguishable from the statistical data operationally
accessible from experiment.  Thus, it is an important problem to determine 
statistical equivalence classes of measuring processes.  
We say that two measuring apparatuses $\bA(\bx)$ and $\bA(\by)$ are
statistically equivalent if for any measuring apparatuses $\bA(\ba)$ and
$\bA(\bb)$, the joint probability distribution of the output variables 
$\ba,\bx,\bb$ and that of $\ba,\by,\bb$ in the successive measurements 
using $\bA(\ba),\bA(\bx),\bA(\bb)$ in this order and 
using  $\bA(\ba),\bA(\by),\bA(\bb)$ in this order
with the same initial state are identical. 
Let $\cI_{\bx},\cI_{\by},\cI_{\ba}$, and $\cI_{\bb}$ be the corresponding 
instruments.  Then,  it is easy to see that if the initial state is $\rh$,
the corresponding joint probability distributions are determined by 
$\|\cI_{\bb}(\De_3)\cI_{\bx}(\De_2)\cI_{\ba}(\De_1)\rh\|$ and 
$\|\cI_{\bb}(\De_3)\cI_{\by}(\De_2)\cI_{\ba}(\De_1)\rh\|$, where $\De_j$ for
$j=1,2,3$ are measurable subsets in their respective value spaces.
It follows that $\bA(\bx)$ and $\bA(\by)$ are statistically equivalent 
if and only if $\cI_{\bx}=\cI_{\by}$.
Thus, the statistical equivalence of measuring processes are naturally
intrduced as follows.

\begin{definition}[Statistical equivalence of measuring processes \cite{84QC}]
Two measuring processes $\M_1=(\cK_1,\si_1,E_1,U_1)$
and $\M_2=(\cK_2,\si_2,E_2,U_2)$ for $(\cM,S)$
are said to be statistically equivalent if their CP instruments $\cI_{\M_1}$ and
$\cI_{\M_2}$ are identical.
\end{definition}

Then, the result of this section can be summarized as follows 
(cf.~\Cite{84QC} [Theorem 5.1]).

\begin{theorem}\label{mainse4}
Let $\cM$ be a von Neumann algebra on a Hilbert space
$\cH$ and $(S,\cF)$ a measurable space.
Then the relation 
\begin{equation}\label{SEMPandCP}
\cI(M,\De)=(\id\otimes\si)[U^*(M\otimes E(\De))U],
\end{equation}
where $\De\in\cF$ and $M\in\cM$,
sets up a one-to-one correspondence between statistical equivalence classes of measuring processes $\M=(\cK,\si,U,E)$
for $(\cM,S)$ and CP instruments $\cI$ for $(\cM,S)$ with  the NEP.
\end{theorem}

\section{Approximations by CP Instruments with  the NEP}
\label{se:approximations}
In the previous section, we have shown that the set of instruments describing measuring 
processes
is characterized by the set $\mathrm{CPInst}_{\mathrm{NE}}(\cM,S)$ of CP instruments
with the NEP.
Since this is a subset of the set $\mathrm{CPInst}(\cM,S)$ of CP instruments, 
it is an interesting problem how large $\mathrm{CPInst}_{\mathrm{NE}}(\cM,S)$ is 
in $\mathrm{CPInst}(\cM,S)$. 
In the case where $\cM$ is a type I factor, or $\cM$ describes a quantum system of finite degrees of freedom, it is known that $\mathrm{CPInst}_{\mathrm{NE}}(\cM,S)=\mathrm{CPInst}(\cM,S)$ holds for any measurable space $(S,\cF)$ \cite{84QC}.
Thus, it is natural to ask if this holds generally. 
In this section, we shall show an affirmative aspect by proving that the equality holds if 
$\cM$ is a direct sum of type I factors, while in the next section we shall show that the
equality does not hold even for a type I von Neumann algebra, which is generally
represented as a direct integral of type I factors.
Thus, the next problem is whether  $\mathrm{CPInst}_{\mathrm{NE}}(\cM,S)$ is
dense in  $\mathrm{CPInst}(\cM,S)$ in an appropriate sense.
In this section, we show a partial affirmative answer with the help of the structure theory of 
von Neumann algebras, in which the injectivity of von Neumann algebras plays a central role.
One of the conditions equivalent to the injectivity of a von Neumann algenbra $\cM$
on a Hilbert space $\cH$
is the existence of a norm one projection of $\B(\cH)$ onto it.
This condition will turn out to be a powerful tool to show that every CP instrument 
defined on an injective von Neumann algebra can be approximated by CP instruments 
with the NEP.  Thus, $\mathrm{CPInst}_{\mathrm{NE}}(\cM,S)$ is dense in
$\mathrm{CPInst}(\cM,S)$ for injective von Neuamann algebras. 
Since in most of physically relevant cases the algebras of local observables are injective,
this result will provide a satisfactory basis for measurement theory of local quantum physics.
In the above results, if we may add, we will put no restriction to measurable spaces 
$(S,\cF)$.

We shall begin with an easier case, where there exists a normal conditional expectation,
or equivalently a normal norm one projection, of $\B(\cH)$ onto $\cM$.
A von Neumann algebra $\cM$ on $\cH$ is said to be atomic
if it is a direct sum of type I factors.
It is known that there exists a normal conditional expectation of
$\B(\cH)$ onto $\cM$ if and only if $\cM$ is atomic
(\Cite{Tak79} [Chapter V, Section 2, Excercise 8]).
In the following theorem we shall show that for any atomic von Neumann algebra $\cM$
the equality $\mathrm{CPInst}_{\mathrm{NE}}(\cM,S)=\mathrm{CPInst}(\cM,S)$
holds, which generalizes the previous result for type I factors (\Cite{84QC}, Theorem 5.1).
\begin{theorem}\label{NCE}
Let $\cM$ be an atomic von Neumann algebra 
and $(S,\cF)$ a measurable space. 
Then, every CP instrument $\cI$ for $(\cM,S)$ has  the NEP.
\end{theorem}

\begin{proof}
Let $\mathcal{E}:\B(\cH)\to\cM$ be a normal conditional expectation.
Let $\cI$ be a CP instrument for $(\cM,S)$.
We define a CP instrument $\widetilde{\cI}$ for $(\B(\cH),S)$ by
\begin{equation}\label{canonical}
\widetilde{\cI}(X,\De)=\cI(\mathcal{E}(X),\De)
\end{equation}
for $\De\in\cF$ and $X\in\B(\cH)$.
For every $\De\in\cF$ and $M\in\cM$,
\begin{equation}
\widetilde{\cI}(M,\De)=\cI(\mathcal{E}(M),\De)=\cI(M,\De).
\end{equation}
It follows from Theorem \ref{NEP} (iii) that $\cI$ has  the NEP.
\end{proof}
We call the CP instrument $\widetilde{\cI}$ for $(\B(\cH),S)$
defined by Eq.~(\ref{canonical}) the $\cE$-canonical extension of $\cI$.
We have $\widetilde{\cI}(\De)=\cE^\ast\circ\cI(\De)\circ e_\cM$ for all $\De\in\cF$,
where the map $e_\cM:\B(\cH)_*\rightarrow \cM_*$ is defined by
$e_\cM(\rh)=\rh|_\cM$ for all $\rh\in\B(\cH)_*$,
and we shall write $\widetilde{\cI}=\cE^\ast\cI$. 

Let $\cM$ be a von Neumann algebra on a Hilbert space $\cH$
and $(S,\cF)$ a measurable space. 
We write $M_\alpha\to^{uw}M$ 
if a net $\{M_\alpha\}$ in $\cM$ ultraweakly converges to an element $M$ of $\cM$.
Let $\cI$ be a CP instrument for $(\cM,S)$
and $\{\cI_\alpha\}$ a net of CP instruments for $(\cM,S)$.
We say that  $\cI_\alpha$ ultraweakly converges to $\cI$
and write $\cI_\alpha\to^{uw}\cI$
if $\cI_\alpha(M,\De)\to^{uw}\cI(M,\De)$ for all $M\in\cM$ and $\De\in\cF$.
In the rest of this section, we shall consider the case where $\cM$ is injective, or equivalently
there exists a (not necessarily normal) norm one projection of $\B(\cH)$ onto $\cM$,
and show that in this case $\mathrm{CPInst}_{\mathrm{NE}}(\cM,S)$ is ultraweakly dense in $\mathrm{CPInst}(\cM,S)$.

We begin with the following proposition useful in the later argument.

\begin{proposition}[Anantharaman-Delaroche \cite{AD95}] \label{AD95}
Let $\cM\subset\cN$ be a pair of von Neumann algebras.  Assume that
there exists a normal faithful semifinite weight $\varphi$ on $\cM$
such that, for all $t\in\mathbb{R}$,
the modular automorphism $\si^\varphi_t$ is induced by a unitary operator in $\cN$.
Then the following conditions are equivalent:

(i) There exists a norm one projection of $\cN$ onto $\cM$.

(ii) There exists a net $\{T_\alpha\}$ of normal CP maps from $\cN$ to $\cM$ such that $T_\alpha(1)\leq 1$ for all $\alpha$ and that $T_\alpha(M)\to^{uw}M$
for all $M\in\cM$.

(iii) There exists a net $\{T_\alpha\}$ of unital normal CP maps from $\cN$ to $\cM$
such that $T_\alpha(M)\to^{uw}M$ for all $M\in\cM$.
\end{proposition}

\begin{proof}
The proof of the equivalence between (i) and (ii) is given in \Cite{AD95}, Corollary 3.9.

(ii)$\Rightarrow$(iii) Let $\{T_\alpha\}$ be a net of normal CP maps from $\cN$ to $\cM$
such that $T_\alpha(1)\leq 1$ for all $\alpha$, and $T_\alpha(M)\to^{uw}M$
for all $M\in\cM$. Choose a normal state $\omega$ on $\cM$.
We define a net $\{T_\alpha'\}$ of unital normal CP maps from $\cN$ to $\cM$ by
\begin{equation}
T_\alpha'(M)=T_\alpha(M)+(1-T_\alpha(1))\omega(M),\hspace{5mm}M\in\cM.
\end{equation}
Then, since $(1-T_\alpha(1))\omega(M)\to^{uw}0$, we have $T'_\alpha(M)\to^{uw}M$ for all $M\in\cM$, and assertion (iii) follows.

(iii)$\Rightarrow$(ii) Obvious.
\end{proof}

The following theorem holds as previously announced.
\begin{theorem}\label{NOP}
Let $\cM$ be an injective von Neumann algebra,
and $(S,\cF)$ a measurable space. For every CP instrument $\cI$ for $(\cM,S)$,
there exists a net $\{\cI_\alpha\}$ of CP instruments with  the NEP such that
$\cI_\alpha\to^{uw}\cI$.
\end{theorem}
\begin{proof}
Suppose that $\cM$ acts on a Hilbert space $\cH$ in a standard form 
without loss of generality.
Then, for every faithful normal state $\varphi$ on $\cM$ and for all $t\in\mathbb{R}$,
the modular automorphism $\si^\varphi_t$ is induced by a unitary operator in
$\B(\cH)$. By Proposition \ref{AD95}, there exists
a net $\{T_\alpha\}$ of unital normal CP maps from $\B(\cH)$ to $\cM$,
such that $T_\alpha(M)\to^{uw}M$ for all $M\in\cM$.
Let $\cI$ be a CP instrument for $(\cM,S)$. 
For every $\alpha$, let $\widetilde{\cI_\alpha}$ be a CP instrument for $(\B(\cH),S)$
defined by
\begin{equation}
\widetilde{\cI_\alpha}(X,\De)=\cI(T_\alpha(X),\De)
\end{equation}
for $\De\in\cF$ and $X\in\cB(\cH)$, 
and let $\cI_\alpha$ be a CP instrument for $(\cM,S)$ obtained by restricting 
$\widetilde{\cI_\alpha}^*$ to $\cM$,
i.e., $\cI_\alpha(M,\De)=\widetilde{\cI_\alpha}(M,\De)$
 for all $\De\in\cF$ and $M\in\cM$.  
Then, it follows from Theorem \ref{NEP} (iii) that $\{\cI_\alpha\}$ 
is a net of CP instruments for $(\cM,S)$ with  the NEP, and it is easy to see that 
$\{\cI_\alpha\}$ ultraweakly converges to $\cI$.
Thus,  the assertion follows. 
\end{proof}

\begin{definition}[Approximately normal extension property]
Let $\cM$ be a von Neumann algebra on a Hilbert space $\cH$,
and $(S,\cF)$ a measurable space.
A CP instrument $\cI$ for $(\cM,S)$ has the approximately normal extension property (ANEP)
if there exists a net $\{\cI_\alpha\}$ of CP instruments with  the NEP such that
$\cI_\alpha\to^{uw}\cI$. We denote by
$\mathrm{CPInst}_{\mathrm{AN}}(\cM,S)$ the set of CP instruments for $(\cM,S)$ with ANEP.  Note that $\mathrm{CPInst}_{\mathrm{AN}}(\cM,S)$ is the ultraweak
closure of $\mathrm{CPInst}_{\mathrm{NE}}(\cM,S)$.
\end{definition}
By definitions of  the NEP and the ANEP, we have
\begin{equation}
\mathrm{CPInst}_{\mathrm{NE}}(\cM,S)\subset\mathrm{CPInst}_{\mathrm{AN}}(\cM,S)
\subset\mathrm{CPInst}(\cM,S).
\end{equation}
By Proposition \ref{NCE} and Theorem \ref{NOP}, more strict relations
among these three sets for two classes of von Neumann algebras are summarized as the following theorem.
\begin{theorem} \label{NEPANEP}
Let $\cM$ be a von Neumann algebra on a Hilbert space $\cH$,
and $(S,\cF)$ a measurable space.
The following statements holds:

(i) $\mathrm{CPInst}_{\mathrm{NE}}(\cM,S)=\mathrm{CPInst}_{\mathrm{AN}}(\cM,S)
=\mathrm{CPInst}(\cM,S)$ if $\cM$ is atomic.

(ii) $\mathrm{CPInst}_{\mathrm{AN}}(\cM,S)=\mathrm{CPInst}(\cM,S)$ if
$\cM$ is injective.
\end{theorem}

We believe that there is a (non-injective) von Neumann algebra $\cM$ such that
$\mathrm{CPInst}_{\mathrm{AN}}(\cM,S)\subsetneq\mathrm{CPInst}(\cM,S)$,
though we are not aware of such an example up to now.
In the next section, we will give examples of CP instruments without  the NEP to show 
the existence of (non-atomic but injective)  von Neumann algebras $\cM$ such that
$\mathrm{CPInst}_{\mathrm{NE}}(\cM,S)\subsetneq\mathrm{CPInst}_{\mathrm{AN}}(\cM,S)$ .

By the way, we have the following stronger results than Theorem \ref{NEPANEP}.
$\mathcal{M}$ is atomic if and only if every CP instrument ${\cI} $ for $(\mathcal{M},S)$ has the 
following property:
There exists a CP instrument $\widetilde{{\cI} }$ for $(\B(\mathcal{H}),S)$
such that $\widetilde{\cI}(M,\De)=\cI(M,\De)$
for all $\De\in\mathcal{F}$ and $M\in\mathcal{M}$
and that $\widetilde{\cI}(X,\De)\in\mathcal{M}$ for all $\De\in\mathcal{F}$
and $X\in \B(\mathcal{H})$.
Similarly, by Proposition \ref{AD95} $\mathcal{M}$ is injective if and only if
for every CP instrument $\cI$ for $(\mathcal{M},S)$
there exists a a net of CP instruments $\{\cI_\alpha\}$ for $(\mathcal{M},S)$
such that $\cI_\alpha\to^{uw}{\cI}$
and that each $\cI_\alpha$ has the above property.
These are mathematically meaningful results but their physical 
significance is not clear to the best of our knowledge.

From a physical point of view, Theorem \ref{NEPANEP} (ii) will provide a satisfactory basis for
 measurement theory of local quantum physics.
A von Neumann algebra $\cM$ is said to be
approximately finite dimensional (AFD) if there exists an increasing net $\{\cM_\alpha\}$
of finite-dimensional von Neumann subalgebras of $\cM$ such that
\begin{equation}
\cM=\overline{\bigcup_\alpha \cM_\alpha}^{uw}.
\end{equation}
It is known that von Neumann algebras describing local observables in quantum field theory
are AFD and separable (i.e.,  with a separable predual) 
under very general postulates, e.g., the Wightman axioms, nuclearity, and asymptotic scale
invariance  \cite{BDF87}. In addition,
a separable von Neumann algebra is injective if and only if it is AFD.
This famous result is established by Connes, Wassermann, Haagerup, Popa and other researchers
\cite{Con94, Tak79}. 
Therefore, the relation
$\mathrm{CPInst}_{\mathrm{AN}}(\cM,S)=\mathrm{CPInst}(\cM,S)$
holds for von Neumann algebras $\cM$ describing observable algebras of most of physically relevant systems;
in those situations every CP instrument on the observable algebra can be considered as a realizable measurement within arbitrarily given error limit $\epsilon>0$, i.e., for every CP instrument $\cI$
and $\phi_j\in\cS(\cM)$, $M_j\in\cM$, and $\De_j\in\cF$ with $j=1,\ldots,n$ 
there exists a measuring process $\bbM$ such that
\beq
|\braket{\phi_j,\cI(M_j,\De_j)}-\braket{\phi_j,\cI_{\bbM}(M_j,\De_j)}|<\epsilon
\eeq
for all $j=1, \ldots, n$.
As above,  accepting that local algebras are AFD and separable, 
we conclude that all CP instruments defined on a local algebra are physically 
realizable within arbitrary error limit.

\section{Existence of a Family of Posterior States and Its Consequences}
\label{se:existence}

In this section, we discuss the existence of a family of posterior states for a given instrument
and normal state. 
We follow notations used in the beginning of section \ref{se:completely}.
In each instance of measurement using a measuring apparatus $\bA(\bx)$
the output variable $\bx$ is assumed to take a value $\bx=s$ in a measurable 
space $(S,\cF)$ independent from how accurately the output variable $\bx$ is read out
by the observer.
Suppose that the measured system $\bS$ is in a state $\rh\in\SM$ just before the measurement
and the measurement leads to the output value $\bx=s$. 
Let $\rh_{\{\bx=s\}}$ be the state just after the measurement.
How is the state $\rh_{\{\bx=s\}}$ determined by the instrument $\cI$ of $\bA(\bx)$?
If $\mathrm{Pr}\{\bx\in\{s\}\|\rh\}>0$, the state $\rh_{\{\bx=s\}}$ should be given by the relation
\begin{equation}\label{eq:posteriori}
\rh_{\{\bx=s\}}=\rh_{\{\bx\in\{s\}\}}=
\frac{\cI(\{s\})\rh}{\|\cI(\{s\})\rh\|}.
\end{equation}
It is, however, impossible to apply this method generally to an arbitrary measuring
apparatus $\bA(\bx)$, since the output probability $\mathrm{Pr}\{\bx\in\De\|\rh\}$ is often
assumed to be continuously distributed.
To resolve this difficulty, being inspired by the concept of conditional probability
in classical probability theory, the concept of a family of posterior states
for any instruments and any normal states was introduced in Refs.~\onlinecite{85CA,85CC}.
Let $\bA(\by)$ be an arbitrary apparatus described by an instrument $\cI'$ for $(\cM,S')$
with measurable space $(S',\cF')$.
Suppose that a measurement using the apparatus $\bA(\bx)$ in a state $\rh$ is
immediately followed by another measurement using the apparatus $\bA(\by)$.
Then, we have the joint probability distribution of $\bx$ and $\by$ defined by
$\Pr\{\by\in\De',\bx\in\De\|\rh\}=\|\cI'(\De')\cI(\De)\rh\|$, where $\De\in\cF$
and $\De'\in\cF'$.
According to classical probability theory, the conditional probability distribution 
$\Pr\{\by\in\De|\bx=s\|\rh\}$
of the output variable $\by$ given the value $\bx=s$ of the output variable $\bx$
is defined by
\beq
\Pr\{\by\in\De',\bx\in\De\|\rh\}=\int_{\De}\Pr\{\by\in\De'|\bx=s\|\rh\}
d\Pr\{\bx\in s\|\rh\}.
\eeq
If the output variable $\bx$ takes the value $\bx=s$, the second measurement must be carried out
on the state $\rh_{\{\bx=s\}}$ and hence the probability distribution of the output 
$\by$ should satisfy the relation
\beql{posterior-conti}
\Pr\{\by\in\De'|\bx=s\|\rh\}=\Pr\{\by\in\De'\|\rh_{\{\bx=s\}}\}.
\eeq
A  family of state $\{\rh_{\{\bx=s\}}\}_{s\in S}$ satisfying
 the above condition 
for any instrument $\cI'$ is called a 
family of posterior states for $(\cI,\rh)$.  
If  $\mathrm{Pr}\{\bx\in\{s\}\|\rh\}>0$, we have
\beq
\Pr\{\by\in\De'|\bx=s\|\rh\}=\Pr\left\{\by\in\De'\left\|\frac{\cI(\{s\})\rh}{\|\cI(\{s\})\rh\|}\right.\right\}
\eeq
and hence \Eq{posterior-conti} generalizes \Eq{posteriori} to the 
continuous case.

In quantum mechanical systems, described by a type I factor, there always exists a 
family of posterior states for any instruments and normal states (\Cite{85CA},
Theorem 4.5).  In contrast to this case, it is already established also in \Cite{85CA}
that its existence for any instruments and normal states is not always guaranteed 
in general quantum systems.  
This has been derived from an interesting conflict between the weak repeatability 
for continuous observables and the existence of a family of posterior states.
In this section, we study an interesting connection between the concept of 
posterior states and the NEP introduced in the previous section, and prove 
that for any CP instrument the NEP is equivalent to the existence of a strongly 
measurable family of posterior states for every normal state.
From this result, we shall also obtain a condition for general instrument to have 
a family of posterior states for every normal state (see Corollary \ref{UNPM}).
On the other hand, we shall give two examples of CP instruments without the NEP,
which arise from weakly repeatable instruments for continuous observables in
a commutative (type I) von Neumann algebra and a type II$_1$ factor.
Thus, we shall conclude that 
in the general case there exists a weakly repeatable CP instrument for a
continuous observable that does not have the corresponding measuring process, 
whereas for separable type I factors every CP instrument has the corresponding measuring 
process but no weakly repeatable instrument exists for continuous observables.

Let $\cM$ be a von Neumann algebra on a Hilbert space $\cH$ and 
$(S,\cF)$ a measurable space.

\begin{definition}
Let $\mu$ be a positive finite measure on $(S,\cF)$.

(i) A family $\{\rh_s\}_{s\in S}$ of (not necessarily normal)
positive linear functionals on $\cM$ is said to be weakly$^*$ $\cF$-measurable 
if the function $s\mapsto\bracket{\rh_s,M}$ is $\cF$-measurable for all $M\in\cM$.

(ii) A family $\{\rh_s\}_{s\in S}$ of (not necessarily normal)
positive linear functionals on $\cM$ is said to be weakly$^*$ $\mu$-measurable 
if the function $s\mapsto\bracket{\rh_s,M}$ is $\mu$-measurable for all $M\in\cM$.

(iii) A family $\{\rh_s\}_{s\in S}$
of normal positive linear functionals
on $\cM$ is said to be strongly ($\cF$-)measurable if
there exists a sequence $\{F_n\}$ of $\cM_\ast$-valued simple functions on $S$ such that
$\lim_{n}\|\rh_s-F_n(s)\|=0$ for all $s\in S$.
\end{definition}

\begin{definition}[Disintegrations and families of posterior states]
Let $\cI$ be an instrument for $(\cM,S)$ and $\rh$ a normal state on $\cM$.
A weakly$^*$ $\Ir$-measurable family $\{\rh_s\}_{s\in S}$ 
of (not necessarily normal) states on $\cM$ is called a disintegration 
with respect to $(\cI,\rh)$ if the relation 
\begin{equation}
\bracket{\cI(\De)\rh,M}=\int_{\De} \bracket{\rh_s,M}d\|\cI(s)\rh\|.
\end{equation}
holds for all $M\in\cM$ and $\De\in\cF$.
A disintegration $\{\rh_s\}_{s\in S}$ with respect to $(\cI,\rh)$
is said to be proper if it satisfies that for any positive $M\in\cM$, 
if $\braket{\rh,\cI(M,S)}=0$ then $\bracket{\rh_s,M}=0$ for all $s\in S$.
A disintegration $\{\rh_s\}_{s\in S}$ with respect to $(\cI,\rh)$
is called a family of posterior states with respect to $(\cI,\rh)$ if
$\rh_s$ is a normal state, i.e., $\rh_s\in\SM$, for all $s\in S$.
\end{definition}

A disintegration $\{\rh_s\}_{s\in S}$ with respect to $(\cI,\rh)$ 
is unique in the following sense:
If $\{\rh'\}_{s\in S}$ is another disintegration with respect to $(\cI,\rh)$,
then $\bracket{\rh_s,M}=\bracket{\rh'_s,M}$ for
$\Ir$-a.e. $s\in S$ and all $M\in\cM$.

From a measurement theoretical point of view, a family $\{\rh_s\}_{s\in S}$ of 
posterior states with respect to $(\cI,\rh)$ is naturally required to 
satisfy the following two conditions.

(i) (Uniqueness) The state $\rh_s$ is uniquely determined by the instrument $\cI$
and the input state $\rh$ with probability one, 
in the sense that if  $\{\rh'_s\}_{s\in S}$ is another family of posterior states,
then we have $\rh_s=\rh'_s$ for $\Ir$-\ae $s\in S$.

(ii) (Integrability) For any $\De\in\cF$, the state $\cI(\De)\rh/\|\cI(\De)\rh\|$
after the measurement in any state $\rh\in\SM$ conditional upon the 
outcome event $\bx\in\De$ is the mixture of all state $\rh_s$ 
with the conditional probability distribution $d\|\cI(s)\rh\|/\|\cI(\De)\rh\|$, i.e., 
the Bochner integral formula
\beq
\frac{\cI(\De)\rh}{\|\cI(\De)\rh\|}=\int_{\De}\rh_s\,
\frac{d\|\cI(s)\rh\|}{\|\cI(\De)\rh\|}
\eeq 
holds for all $\De\in\cF$.

If a family $\{\rh_s\}_{s\in S}$ of posterior states
with respect to $(\cI,\rh)$ is strongly $\cF$-measurable,
then the $\SM$-valued function $s\mapsto \rh_s$ is Bochner integrable
with respect to every probability measure on $(S,\cF)$.  
In addition, for two strongly $\cF$-measurable families
$\{\rh_s\}_{s\in S}$, $\{\rh'_s\}_{s\in S}$ of posterior states
with respect to $(\cI,\rh)$, it holds that
$\rh_s=\rh'_s$ for $\Ir$-a.e. $s\in S$.
Thus, a family of posterior states $\{\rh_s\}_{s\in S}$  
satisfies the above two conditions, (i) and (ii),
if and only if it is strongly measurable.
Thus, the most physically relevant concept to describe the state 
after the measurement conditional upon the value of the output variable
is considered to be defined as a strongly measurable family of posterior states.

The following results were obtained in the previous investigations.

\begin{theorem}[\protect{\Cite{85CC}, Theorem 4.3}]
Let $\cM$ be a von Neumann algebra on a Hilbert space $\cH$
and $(S,\cF)$ a measurable space.
For any instrument $\cI$ for $(\cM,S)$ and normal state $\rh$ on $\cM$,
a proper disintegration $\{\rh_s\}_{s\in S}$
with respect to $(\cI,\rh)$ always exists.
\end{theorem}

\begin{theorem}[\protect{\Cite{85CA}}, Theorem 4.5]
Let $\cH$ be a Hilbert space and  $(S,\cF)$ a measurable space.
For any CP instrument $\cI$ for $(\BH,S)$ and normal state $\rh$ on $\BH$,
a strongly measurable family $\{\rh_s\}_{s\in S}$ of posterior states 
with respect to $(\cI,\rh)$ always exists.
\end{theorem}

The next example shows that not all CP instruments defined on injective von Neumann algebras have  the NEP, 
and is strongly related to Theorems \ref{NEPAFAPS} and \ref{WeakRepeatable}
below.

\begin{example}[\protect{\Cite{85CA}, pp.~292--293}] \label{ex85}
Let $m$ be Lebesgue measure on $[0,1]$ and $\cI$ a CP instrument for
$(L^{\infty}([0,1],m), [0,1])$ defined by $\cI(f,\De)=[\chi_\De] f$.
for all $\De\in\mathcal{B}([0,1])$ and $f\in L^{\infty}([0,1],m)$.
Let $\rh=m\in \cS_n(L^{\infty}([0,1],m))$, i.e., 
$\bracket{\rh,[f]}=\int_{0}^{1}f(x) dx$ for all $f\in M^{\infty}([0,1],m)$.
Then, 
there exists no family $\{\rh_x\}_{x\in [0,1]}$ of
posterior states with respect to $(\cI,\rh)$.
$L^{\infty}([0,1],m)$ is an injective (maximal abelian) von Neumann subalgebra
of the von Neumann algebra $\B(L^2([0,1],m))$
of bounded operators on a separable Hilbert space $L^2([0,1],m)$.
It is well-known that there is no normal conditional expectation of
$\B(L^2([0,1],m))$ onto $L^{\infty}([0,1],m)$.
\end{example}

By the above example and Theorem \ref{NEPANEP}, we have
\begin{equation}
\mathrm{CPInst}_{\textrm{NE}}(L^{\infty}([0,1],m),[0,1])
\subsetneq \mathrm{CPInst}(L^{\infty}([0,1],m),[0,1]).
\end{equation}

Let $(S,\cF,\mu)$ be a finite measure space.
Let $E$ be a Banach space.
Denote by $\cL^{1}(S,\mu,E)$ the space of Bochner $\mu$-integrable 
$E$-valued functions on $S$.
A function $f\in \cL^{1}(S,\mu,E)$ is called strongly 
$\mu$-negligible if $f(s)=0$ holds for $\mu$-\ae $s\in S$.
The quotient space of $\cL^{1}(S,\mu,E)$ modulo the strongly $\mu$-negligible functions, 
denoted by $L^{1}(S,\mu,E)$, is a Banach space with the $L^1$ norm defined by 
$\|[f]\|_{1}=\int_S \|f(s)\| d\mu(s)$.

Let $\cM$ and $\cN$ are von Neumann algebras on Hilbert spaces $\cH$ and $\cK$, respectively.
Then, the algebraic tensor product $\cM\aotimes\cN$ can be defined on the tensor product
Hilbert space $\cH\otimes\cK$.  The uniform norm closure of the *-algebra 
$\cM\aotimes\cN$ on $\cH\otimes\cK$ is *-isomorphic with the injective 
C*-tensor product $\cH\motimes \cK$ and the weak closure of  $\cM\aotimes\cN$
on $\cH\otimes\cK$ is *-isomorphic with the W*-tensor product $\cM\ootimes\cN$.
The predual of  $\cM\ootimes\cN$ is isometrically isomorphic to the 
Banach space tensor product $\cM_{*}\otimes_{\min^{*}}\cN_{*}$ of the preduals 
$\cM_{*}$ and $\cN_{*}$ with the adjoint cross norm to the injective C*-cross
norm $\|\cdot\|_{\min}$.  For the case where $\cN=L^{\infty}(S,\mu)$,  we have
$\cM\motimes L^{\infty}(S,\mu)=\cM\otimes_{\la}L^{\infty}(S,\mu)$,
where $\la$ stands for the least cross norm, and $L^{\infty}(S,\mu)_{*}=L^{1}(S,\mu)$.
Thus, we have $(\cM\ootimes L^{\infty}(S,\mu))_{*}\cong
\cM_{*}\otimes_{\ga}L^{1}(S,\mu)$, where $\ga$ stands for the greatest cross norm,
and by the Grothendieck theorem \cite{Gro55} we have 
$\cM_{*}\otimes_{\ga}L^{1}(S,\mu)\cong L^{1}(S,\mu,\cM_{*})$.
Thus, we have the following.

\begin{theorem}[\protect{\Cite{Sak71}}, Proposition 1.22.12]
\label{Sakai}
Let $\cM$ be a von Neumann algebra and $(S,\cF,\mu)$ a finite measure space.  
Then, the relation
\beql{Sakai}
\bracket{\ph,M\otimes [g]}=\int_{S}g(s)\bracket{f_{\ph}(s),M}\,d\mu(s),
\eeq
where $g\in M^{\infty}(S,\mu)$ and $M\in\cM$,
sets up an isometric isomorphism $\ph\mapsto [f_{\ph}]$
of $( \cM\ootimes L^{\infty}(S,\mu))_{*}$ onto $L^{1}(S,\mu,\cM_*)$.
\end{theorem}

For further information about vector-valued integrals and tensor products of operator
algebras we refer the reader to \Cites{TT69,Sak71,Tak79}.  

\begin{proposition} \label{Measurable}
Let $\cM$ be a von Neumann algebra on a Hilbert space $\cH$
and $(S,\cF)$ a measurable space.
For every CP instrument $\cI$ for $(\cM,S)$ with  the NEP and normal state $\rh$
on $\cM$, a strongly measurable
family $\{\rh_s\}_{s\in S}$ of posterior states with respect to $(\cI,\rh)$ always exists.
\end{proposition}
\begin{proof}
Let $\cI$ be a CP instrument for $(\cM,S)$ with the NEP.
Let $\vp$ be a faithful normal state on $\cM$.
Let  $\mu$ and $\nu$ be the probability measures on $(S,\cF)$ defined by
$\mu=\|\cI\vp\|$ and $\nu=\|\cI\rh\|$.   Then, it is easy to see that $\nu\ll\mu$,
so that there is $\Ga\in\cF$ such that $\nu\equiv\mu_\Ga$, where
$\mu_\Ga(\De)=\mu(\De\cap\Ga)$  for all $\De\in\cF$
(\Cite{Hal51}, Theorem 47.~2), and hence $L^{\infty}(S,\nu)$ is naturally identified with 
the direct summand $L^{\infty}(\Ga,\mu)$ of $L^{\infty}(S,\mu)$ by
the correspondence $\De\mapsto \De\cap\Ga$.
By the NEP there exists a unital normal CP map 
$\Phi:\cM\ootimes L^{\infty}(S,\mu)\to\cM$ such that 
$\Phi(M\otimes [\chi_\De])=\cI(M,\De)$ for all $M\in\cM$ and $\De\in\cF$. 
Let $\Phi_{\Ga}$ be the unital normal CP map 
$\Phi_{\Ga}:\cM\ootimes L^{\infty}(S,\nu)\to\cM$ defined by
$\Phi_{\Ga}(M\otimes [\chi_\De])=\Phi(M\otimes [\chi_\Ga\chi_\De])$
for all $M\in\cM$ and $\De\in\cF$.
By the normality of $\Phi_{\Ga}$ there exists the predual map  
$(\Phi_{\Ga})_*:\cM_*\to(\cM\ootimes L^{\infty}(S,\nu))_{*}$ of $\Phi_{\Ga}$. 
Hence, we have $(\Phi_{\Ga})_*\rh\in (\cM\ootimes L^{\infty}(S,\nu))_{*}$.
Since $\|\cI(S\setminus \Ga)\rh\|=\nu(S\setminus\Ga)=0$, we have 
$\cI(\De\cap\Ga)\rh=\cI(\De)\rh$ for all $\De\in\cF$, and hence we have
$\bracket{(\Phi_{\Ga})_*\rh,M\otimes [\ch_\De]}=\bracket{\cI(\De)\rh,M}$
for all $\De\in\cF$ and $M\in\cM$.  
Thus, by Theorem \ref{Sakai} there exists $f_\rh\in \cL^{1}(S,\nu,\cM_*)$ such that
\beq
\bracket{\cI(\De)\rh,M}
=\int_{\De}\bracket{f_{\rh}(s),M}\,d\nu(s)
\eeq
for all $\De\in\cF$ and $M\in\cM$.  It follows that the function 
$s\mapsto\bracket{f_{\rh}(s),M}$ is the Radon-Nikodym derivative of the signed 
measure $\bracket{\cI\rh,M}$ with respect to $\nu$, i.e., 
$$
\bracket{f_{\rh}(s),M}=\frac{d\bracket{\cI\rh,M}}{d\nu}(s).
$$
By the properties of Radon-Nikodym derivative, 
there exists $N\in\cF$ such that $\nu(N)=0$ and that $f_{\rh}(s)\ge 0$ and 
$\bracket{f_{\rh}(s),1}=1$ hold for all $s\in S\backslash N$.  Now, we define 
$\rh_{s}= f_{\rh}(s)$ if $s\in S\backslash N$ and $\rh_{s}=\cI(S)\rh$ if $s\in N$.
Then, it is easy to see that $\{\rh_s\}$ is a strongly measurable family 
of posterior states for $(\cI,\rh)$.
\end{proof}

Proposition \ref{Measurable}
states that there exists a family of posterior states for $(\cI,\rh)$
for every normal state $\rh$ on $\cM$ if $\cI$ has  the NEP. 
The converse is also true.

\begin{theorem} \label{NEPAFAPS}
Let $\cM$ be a von Neumann algebra on a Hilbert space $\cH$,
$(S,\cF)$ a measurable space,
and $\cI$ a CP instrument for $(\cM,S)$.
The following conditions are equivalent:

(i) $\cI$ has  the NEP.

(ii) For every normal state $\rh$ on $\cM$,
there exists a strongly $\cF$-measurable family
$\{\rh_s\}_{s\in S}$ of posterior states with respect to $(\cI,\rh)$.

\end{theorem}
\begin{proof}

(i)$\Rightarrow$(ii) This part has been given as Proposition \ref{Measurable}.

(ii)$\Rightarrow$(i) 
Let $\cI$ be a CP instrument  for $(\cM,S)$ satisfying assumption (ii)
and $\Ph_{\cI}:\cM\otimes_{\alg} L^{\infty}(S,\cI)\to\cM$ its linear extension.  
In order to show (i) holds, it suffices to show that for any normal state 
$\rh$ on $\cM$ the linear functional $\rh\circ\Ph_{\cI}$ is ultraweakly continuous on 
$\cM\otimes_{\alg} L^{\infty}(S,\cI)$. Then,  $\Ph_{\cI}$ is ultraweakly
continuous on $\cM\otimes_{\alg} L^{\infty}(S,\cI)$ and uniquely extends to
$\widetilde{\Ps}_{\cI}: \cM\ootimes L^{\infty}(S,\cI)\to\cM$
to conclude (i).
Let $\vp$ be a normal faithful state on $\cM$ and $\mu=\|\cI\vp\|$.
Then, $L^{\infty}(S,\cI)$ is *-isomorphic with 
$L^{\infty}(S,\mu)$
and $(\cM\ootimes L^{\infty}(S,\mu))_{*}$ is isomorphic with 
$L^{1}(S,\mu,\cM_{*})$ by Theorem \ref{Sakai}.
Let $\rh\in\SM$ and $\nu=\|\cI\rh\|$.  Then there exists a strongly $\cF$-measurable
family $\{\rh_s\}_{s\in S}$ of posterior states with respect to $(\cI,\rh)$.
By definition, $\rh_s\in L^1(S,\mu,\cM_*)$.
In addition, it holds that $\nu\ll\mu$.
There exists a non-negative 
$\cF$-measurable function $\lambda$ on $S$ such that
$\lambda=d\nu/d\mu$, $\mu$-a.e.
The family $\{\la(s)\rh_s\}_{s\in S}$ is strongly $\cF$-measurable, and 
$$
\int_S \| \la(s)\rh_s\|d\mu(s)=\int_S \la(s)\|\rh_s\|d\mu(s)
 = \int_S \dfrac{d\nu}{d\mu}(s)d\mu(s)=1.\nonumber
$$
Thus, $\{\la(s)\rh_s\}_{s\in S}\in L^1(S, \mu,{\cM_*})$. 
By Theorem \ref{Sakai} there exists an element 
$\ph\in(\cM\ootimes L^{\infty}(S,\mu))_{*}$
such that the function $f_{\ph}(s)=\la(s)\rh_s$ satisfies \Eq{Sakai}.
Let $M\otimes [\ch_{\De}]\in \cM\otimes_{\rm alg}L^{\infty}(S,\mu)$,
where $\De\in\cF$ and $M\in\cM$.
By the definition of a family of posterior states, we have
\beqas
\bracket{\rh\circ\Ph_{\cI},M\otimes[\ch_{\De}]}
&=&
\bracket{\rh,\cI(M,\De)}
=\int_{\De}\bracket{\rh_{s},M}d\nu(s)\\
&=&
\int_{\De}\bracket{\la(s)\rh_s,M}d\mu(s)
=\bracket{\ph,M\otimes[\ch_{\De}]}
\eeqas
Thus, $\rh\circ\Ph_{\cI}$ coincides with an ultraweakly continuous linear functional
on $\cM\ootimes L^{\infty}(S,\mu)$,
and the assertion follows.
\end{proof}

By the proof of Theorem \ref{NEPAFAPS}, 
we see that the following holds for (not necessarily CP) instruments
and gives a condition equivalent to the existence of a family of posterior states
for every normal state.

\begin{corollary}\label{UNPM}
Let $\cM$ be a von Neumann algebra on a Hilbert space $\cH$,
$(S,\cF)$ a measurable space,
and $\cI$ an instrument for $(\cM,S)$.
The following conditions are equivalent:

(i)
 There exists a unital normal positive map $\Phi_\cI:\cM \ootimes  
L^{\infty}(S,\cI)\to\cM$ such that
\begin{equation}
\cI(M,\De)=\Phi_\cI(M\otimes [\chi_\De])
\end{equation}
for all $M\in\cM$ and $\De\in\cF$.

(ii)
 For every normal state $\rh$ on $\cM$,
there exists a strongly measurable family
$\{\rh_s\}_{s\in S}$ of posterior states with respect to $(\cI,\rh)$.
\end{corollary}

In the following, we discuss the repeatability of instruments.
The repeatability hypothesis as a general principle has been abandoned,
but the class of instruments satisfying the repeatability is still worth reconsidering.
We shall give a condition for repeatable CP instruments to have  the NEP.

\begin{definition}[Repeatability, weak repeatability and discreteness] Let $\cM$ be
a von Neumann algebra and $(S,\cF)$ a measurable space.

(i) An  instrument $\cI$ for $(\cM,S)$ is said to be repeatable
if it satisfies $\cI(\De_1)\cI(\De_2)=
\cI(\De_1\cap\De_2)$ for all $\De_1,\De_2\in\cF$.

(ii) An instrument $\cI$ for $(\cM,S)$ is said to be weakly repeatable
if it satisfies $\cI(\cI(1,\De_2),\De_1)=
\cI(1,\De_1\cap\De_2)$ for all $\De_1,\De_2\in\cF$.

(iii) An instrument $\cI$ for $(\cM,S)$ is said to be discrete if
there exists a countable subset $S_0$ of $S$ and a map $T:S_0\mapsto P(\cM_{*})$
such that
\begin{equation}
\cI(\De)=\sum_{s\in\De}T(s)
\end{equation}
for all $\De\in\cF$.
\end{definition}

It is obvious that every repeatable instrument is weakly repeatable.
\begin{remark}
Suppose that $(S,\cF)$ is a standard Borel space.
In \Cite{85CA}, an instrument $\cI$ for $(\cM,S)$ is said to be discrete
if there exists a countable subset $S_0$ of $S$ such that $\cI(S\setminus S_0)=0$.
For any standard Borel spaces, two definitions of discreteness are equivalent.
The definition of discreteness in this paper is a natural generalization of that for measures
\cite{KF61}.
\end{remark}

Davies and Lewis  \cite{DL70} conjectured that every repeatable instrument 
for $\B(\cH)$ is discrete.
This conjecture was affirmatively resolved for CP instruments in \Cite{84QC}, Theorem 6.6,
and for the general case in \Cite{85CA}, Theorem 5.1 as follows.

\begin{theorem}[\protect{\Cite{85CC}, Theorems 5.1 and 5.2}]
\label{Ozawa-discreteness}
Let $\cM$ be a von Neumann algebra, $(S,\cF)$ be a standard Borel space.
and $\cI$ a weakly repeatable instrument for $(\cM,S)$.
If for a faithful normal state $\vp$ on $\cM$ there is a family $\{\rh_s\}_{s\in S}$ 
of posterior states with respect to $(\cI,\vp)$, then $\cI$ is discrete.
In particular, for a separable Hilbert space $\cH$, every weakly
repeatable instrument $\cI$ for $(\cB(\cH),S)$ is discrete.
\end{theorem}

We shall strengthen the former result to arbitrary CP instruments with the NEP
on a standard Borel space by using the above theorem and the method in this section 
as shown below.

Let $(S,\cF)$ be a measurable space and $S_0$ be a countable subset of $S$.
We define a binary relation $\sim$ on $S_0$ by $s_1\sim s_2$ if $\{s_1,s_2\}\subset\De$
or  $\{s_1,s_2\}\subset\De^{c}$
for every $\De\in\cF$.  Denote by $[s]$ the equivalence class of $s\in S_0$, i.e., 
$[s]=\{s'\in S_0 | s'\sim s\}$.

\begin{lemma} \label{MMP}
Let $(S,\cF)$ be a measurable space and $S_0$ a countable subset of $S$.
There is an at most countable, mutually disjoint family 
$\{\os\}_{[s]\in S_0/\sim}\subset\cF$ such that $[s]\subset \os$ for all $s\in S_0$.
\end{lemma}

\begin{proof}
Let $R=\{(s_1,s_2)\in S_0\times S_0 | s_1\sim s_2\}$.
By the definition of the equivalence $\sim$ on $S_0$,
there is a map $F:(S_0\times S_0)\setminus R\to \cF$ such that
$s_1\in F(s_1,s_2)$, $s_2\notin F(s_1,s_2)$ and $F(s_1,s_2)=F(s_2,s_1)^c$
for every $(s_1,s_2)\in (S_0\times S_0)\setminus R$. For every $s\in S_0$, we define
$\os\in\cF$ by
\begin{equation}
\os=\bigcap_{s_1\sim s} \bigcap_{s_2\not \sim s_1} F(s_1,s_2).
\end{equation}
Then, this $\{\os\}_{[s]\in S_0/\sim}$ is the desired family.
\end{proof}

\begin{proposition} \label{Discrete}
Let $\cM$ be a von Neumann algebra on a Hilbert space $\cH$
and $(S,\cF)$ a measurable space.
Every discrete CP instrument $\cI$ for $(\cM,S)$ has  the NEP.
\end{proposition}

\begin{proof}
Let $\cI$ be a discrete CP instrument for $(\cM,S)$ such that 
$\cI(\De)=\sum_{s\in\De}T(s)$, where $T$ is a $P(\cM_{*})$-valued map on 
a countable set $S_0\subset S$. By Lemma \ref{MMP},
there is an at most countable, mutually disjoint family
 $\{\os\}_{[s]\in S_0/\sim}\subset\cF$ such that
$[s]\subset \os$ for all $s\in S_0$.
Then, a strongly measurable family $\{\rh_s\}_{s\in S}$ of posterior 
states with respect to $(\cI,\rh)$ is defined by
$$
\rh_s=
\left\{
\begin{array}{ll}
\dfrac{\cI(\os)\rh}{\|\cI(\os)\rh\|},
& \mbox{if $\|\cI(\os)\rh\|>0$,}\\
\cI(S)\rh,& \mbox{otherwise.}\nonumber
\end{array}
\right.
$$
By Theorem \ref{NEPAFAPS}, $\cI$ has  the NEP.
\end{proof}

\begin{theorem} \label{WeakRepeatable}
Let $\cM$ be a von Neumann algebra on a Hilbert space $\cH$
and $(S,\cF)$ a standard Borel space.
A weakly repeatable CP instrument $\cI$ for $(S,\cF)$ has  the NEP
if and only if it is discrete.
\end{theorem}

\begin{proof}
By Proposition \ref{Discrete}, discreteness implies  the NEP.
The converse follows from Theorem \ref{NEPAFAPS} and Theorem 
\ref{Ozawa-discreteness}.
\end{proof}

By Corollary \ref{UNPM}, the proof of Proposition \ref{Discrete} and 
Theorem \ref{Ozawa-discreteness},
we have the following.

\begin{corollary}
Let $\cM$ be a von Neumann algebra on a Hilbert space $\cH$
and $(S,\cF)$ a standard Borel space.
For a weakly repeatable instrument $\cI$ for $(\cM,S)$,
the following conditions are equivalent:

(i)
 There exists a unital normal positive map 
 $\Ph_\cI:\cM\ootimes L^{\infty}(S,\cI)\to\cM$ such that
\begin{equation}
\cI(M,\De)=\Ph_\cI(M\otimes [\chi_\De])
\end{equation}
for all $M\in\cM$ and $\De\in\cF$.

(ii)
 For every normal state $\rh$ on $\cM$,
there exists a strongly measurable family $\{\rh_s\}_{s\in S}$ of
posterior states with respect to $(\cI,\rh)$.

(iii)
 For a normal faithful state $\vp$ on $\cM$,
there exists a family $\{\vp_s\}_{s\in S}$ of
posterior states with respect to $(\cI,\vp)$.

(iv) $\cI$ is discrete.
\end{corollary}

We shall show another example of a CP instrument without  the NEP in addition to Example \ref{ex85}.

\begin{example}\label{example-2}
Let $\cM$ be an AFD von Neumann algebra of type $\mathrm{II}_1$
on a separable Hilbert space $\cH$.
Let $A$ be a self-adjoint operator with continuous spectrum affiliated with $\cM$,
$E^A$ the spectral measure of $A$,
and $\mathcal{E}$ a normal conditional expectation of $\cM$ onto
$\{A\}'\cap\cM$,
where $\{A\}'=\{E^{A}(\De)\mid \De\in \cB(\R)\}'$; the existence of a normal 
conditional expectation of a $\si$-finite and finite von Neumann algebra
onto its subalgebras was first found by Umegaki \cite{Ume54}.
We define a CP instrument $\cI_A$ for $(\cM,\R)$ by 
\beql{repeatable}
\cI_A(M,\De)=\mathcal{E}(M)E^A(\De)
\eeq
for all $M\in\cM$ and $\De\in\mathcal{B}(\mathbb{R})$.
Then, $\cI_A$ is not discrete by the continuity of the spectrum of $A$,
and by the property of conditional expectation $\cI_A$ is (weakly) repeatable
as discussed by Davies-Lewis \cite[Theorem 9]{DL70}.
Hence it does not have the NEP but has ANEP by the injectivity of $\cM$.
\end{example}

Instruments and measuring processes play different but a sort of complementary roles
in quantum measurement theory.  Von Neumann \cite{vN32} introduced the repeatability
hypothesis solely from a statistical requirement extracted from the Compton-Simmons 
experiment, and derived the famous measurement-induced state change rule, called 
the projection postulate, for non-degenerate observables, which was eventually 
extended to degenerate observables by L\"{u}ders \cite{Lud51}.  
Then, von Neumann \cite{vN32} raised the problem as to the consistency 
of the projection postulate with fundamental postulates for the standard quantum mechanics. 
To solve this problem, von Neumann introduced quantum mechanical description 
for process of measurement by the interaction, consistent with
the Schr\"{o}dinger equation, between the object and the apparatus as well as
by the subsequent direct measurement,
consistent with the Born rule,  of the meter observable in the apparatus, and showed 
that the state change described by the projection postulate can be consistently described 
by such a  description of the process of measurement.

In out attempt to local quantum measurement theory, the above scenario
has been ultimately generalized as the representation theorem (Theorem \ref{mainse4})
and the density theorem (Theorem \ref{NEPANEP}) for CP instruments 
with the NEP.   The representation theorem states that every CP instruments with the NEP,
as a statistical description of measurement, is consistent with the dynamical description 
represented by a measuring process.  The density theorem states that every CP instrument
on an injective von Neumann algebra can be realized by a measuring process within arbitrary
error limits.  Note that in local quantum physics, a local algebra $\cM$ is broadly shown 
to be an AFD von Neumann algebra on a separable Hilbert space, so that $\cM$ is injective,
An interesting aspect of this new scenario is to allow a more flexible approach to the
repeatability hypothesis for continuous observables.  By Theorem \ref{WeakRepeatable}
no weakly repeatable instrument for a continuous observable has the corresponding 
measuring process.
However, accepting that  local algebras are injective,  a continuous observable
affiliated with a local algebra may have a weakly repeatable CP instrument 
as defined by \Eq{repeatable} and this instrument is considered to be realizable 
within arbitrary error limit by the density theorem.
This is a strong contrast to measurement theory for quantum systems with finite degrees 
of freedom, in which no weakly repeatable instruments for continuous observables 
exists.

In the present and the preceding sections,  
we have developed the theory of CP instruments defined on general von Neumann algebras,
and greatly deepened our understanding of measurement in quantum systems 
with infinite degrees of freedom.
Especially, the normal extension property introduced in Section \ref{se:completely} 
plays a decisive role, which was shown to be equivalent to the existence of a measuring process 
and that of a strongly measurable family of
posterior state for every normal state. Furthermore, we established that
all CP instruments defined on von Neumann algebras describing most of physical systems 
have ANEP.  We finally apply our method to local measurements in AQFT in the next section.

\section{DHR-DR Theory and Local Measurement}
\label{se:dhr-dr}
First, we shall list assumptions for algebraic quantum field theory.
We refer readers to Refs.~\onlinecite{Ara00,Haa96} for standard references 
on algebraic quantum field theory (AQFT).
In AQFT  it is taken for granted that we can  measure only elements of  the 
set ${\cA} (\cO)$ of observables
in bounded regions $\cO$ of four-dimensional Minkowski space
and that each ${\cA} (\cO)$ forms an operator algebra.
Thus, in AQFT, we approach the nature of quantum fields exactly through the family
$\{{\cA} (\cO)\}_{\cO\in\cK}$ of observable algebras.
The purpose of AQFT is to select families of 
operator algebras suitable for the description of quantum fields. 
We introduce the concept of local nets of observables as follows.

\begin{description}
 \item[1 (Local net)] 
Let $\{{\cA} (\cO)\}_{\cO\in\cK}$ be a family of $W^*$-algebras over a causal poset $\cK$ of bounded subregions of
four-dimensional Minkowski space $(\mathbb{R}^4,\eta)$,
where $\eta$ is the Minkowski metric on $\mathbb{R}^4$, satisfying the following four conditions:

(i) $\cO_1\subset \cO_2\in\cK$ $\Rightarrow$
${\cA} (\cO_1)\subset{\cA} (\cO_2)$.

(ii) if $\cO_1$ and $\cO_2$ are causally separated from each other,
then ${\cA} (\cO_1)$ and ${\cA} (\cO_2)$ mutually commute.

(iii) $\bigcup_{\cO\in\cK} {\cA} (\cO)$ is a dense *-subalgebra of
a C*-algebra ${\cA} $.

(iv) there is a strongly continuous automorphic action $\alpha$ on ${\cA} $
of the Poincare group $\mathcal{P}_+^{\uparrow}$ such that $\alpha_{(a,L)}({\cA} (\cO))={\cA} (L\cO+a)=
{\cA} (k_{(a,L)}\cO)$
for any $g=(a,L)\in\mathcal{P}_+^{\uparrow}=
\mathbb{R}^4\rtimes \cL_+^{\uparrow}$ and $\cO\in\cK$, where $\cL_+^{\uparrow}$ is the Lorentz group and 
$k_g:\mathbb{R}^4\rightarrow \mathbb{R}^4$ is defined 
for every $g=(a,L)\in\mathcal{P}_+^{\uparrow}$
by $k_{(a,L)}x=L x+a$ for all $x\in \mathbb{R}^4$.
\end{description}

We call a family $\{{\cA} (\cO)\}_{\cO\in\cK}$ of
W*-algebras satisfyng the above conditions a (W*-)local net of observables.

Next, we shall consider physical states and representations of ${\cA} $
in the case where a vacuum is fixed as a reference state%
\footnote{One may consider that we can prepare states only in bounded local regions,
i.e., (normal) states on ${\cA} (\cO)$ for some $\cO\in\cK$.
This intuition is realized as the concept of local state \cite{OOS15}.
A local state $T$ on ${\cA} $ with an inclusion pair of regions
$(\cO_1,\cO_2)$
 is defined as a unital CP map $T$ on ${\cA} $ satisfying the following conditions,
where $\cO_1,\cO_2\in\cK$ such that $\cO_1\subsetneq\cO_2$:
(i) $T(AB)=T(A)B$ for all $A\in{\cA} $ and $B\in{\cA} ((\cO_2)')$;
(ii) There exists $\varphi\in{\cA} (\cO_1)_{\ast,1}$ such that
$T(A)=\varphi(A)1$ for all $A\in{\cA} (\cO_1)$.
On the other hand, we need states on ${\cA} $ to describe the given physical system.
By using local states, states on ${\cA} $ can be regarded as local states with $(\cO_1,\cO_2)$
in the limit of $\cO_1$ tending to the whole space $\mathbb{R}^4$.
The authors believe that the use of local states would be helpful for developing measurement theory
in local quantum physics in the future.\\ \\ }.
In the setting of algebraic quantum field theory,
it is assumed that all physically realizable states on ${\cA} $
and representations of ${\cA} $ are locally normal, i.e.,
normal on ${\cA} (\cO)$ for all $\cO\in\cK$.
One of the most typical reference states is a vacuum state,
which is a state on lowest every in some coordinate (\Cite{Ara00}, Definition 4.3).
For simplicity, we define a vacuum state $\omega_0$ as follows
(see also \Cite{Ara00}, Theorem 4.5 and \Cite{Haa96} for details).

\begin{description}
 \item[2 (Vacuum state and representation)] 
A vacuum state $\omega_0$ is a $\mathcal{P}_+^{\uparrow}$-invariant locally normal
pure state on ${\cA} $. We denote by $(\pi_0,\cH_0,U,\Omega)$ the GNS representation of
$({\cA} ,\mathcal{P}_+^{\uparrow},\alpha,\omega_0)$.
In addition, it is assumed that the spectrum of the generator $P=(P_\mu)$ of the translation part of $U$
is contained in the closed future lightcone $\overline{V_+}$.
\end{description}

For every $\cO\in\cK$, we denote by $\overline{\cO}$ the closure of $\cO$
and define the causal complement $\cO'$ of $\cO$
by 
\beq
\cO'=\{x\in\mathbb{R}^4\;|\;\eta(x-y,x-y) (=(x-y)^2)<0, y\in\cO\}.
\eeq
For $\cO_1,\cO_2\in\cK$, we denote by $\cO_1\Subset \cO_2$
whenever $\overline{\cO_1}\subsetneq \cO_2$.
We denote by $\cK^{DC}$ the subset of $\cK$ consisting of double cones, i.e.,
\beq
\cK^{DC}=\{(a+V_+)\cap(b-V_+)\in\cK\;|\;a,b\in\mathbb{R}^4\}.
\eeq
Furthermore, we adopt the following notations:
\begin{align}
\cK_\Subset=\{&(\cO_1,\cO_2)\in\cK\times\cK\;|\;
\cO_1\Subset \cO_2\},\\
\cK_\Subset^{DC}=\{&(\cO_1,\cO_2)\in \cK_\Subset\;|\;
\cO_1\; \text{and}\; \cO_2\;\text{are}\;\text{double}\;\text{cones}\}.
\end{align}
For a local net $\{{\cA} (\cO)\}_{\cO\in\cK}$ and
a vacuum state $\omega_0$ on ${\cA} $, we assume the following three conditions:

\begin{description}
 \item[A (Property B)] 
$\{{\cA} (\cO)\}_{\cO\in\cK}$ has property B, i.e., 
for every pair $(\cO_1,\cO_2)\in \cK_\Subset$ of regions
and projection operator $E\in{\cA} (\cO_1)$, there is an isometry operator
$W\in{\cA} (\cO_2)$ such that $WW^*=E$ and $W^* W=1$.

 \item[B (Haag duality)] 
We define the dual net $\{{\cA} ^d(\cO)\}_{\cO\in\cK^{DC}}$
of $\{{\cA} (\cO)\}_{\cO\in\cK}$ with respect to the vacuum representation $\pi_0$
by ${\cA} ^d(\cO)=\pi_0({\cA} (\cO'))'$
for all $\cO\in\cK^{DC}$, where ${\cA} (\cO')=
\overline{\bigcup_{\cO_1\in\cK,\cO_1\subset\cO'}
{\cA} (\cO_1)}^{\|\cdot\|}$.
Then, 
$\{{\cA} (\cO)\}_{\cO\in\cK}$ satisfies Haag duality in $\pi_0$,
i.e., 
${\cA} ^d(\cO)=\pi_0({\cA} (\cO)){''}$
for all $\cO\in\cK^{DC}$.

 \item[C (Separability)] $\cH_0$ is separable.
\end{description}

In the case where a local net $\{{\cA} (\cO)\}_{\cO\in\cK}$ and
a vacuum state $\omega_0$ on ${\cA} $ are fixed and satisfy the above conditions,
we are in a typical situation appearing in
the Doplicher-Haag-Roberts and  Doplicher-Roberts theory (DHR-DR theory, for short),
which selects a local excitations.
A well-known condition selecting physical representations which describe local excitations
is called the DHR selection criterion.
A representation $\pi$ of ${\cA} $ on a Hilbert space $\cH$
is said to satisfy the DHR selection criterion
in support with a bounded region $\cO$ if
the restriction $\pi|_{{\cA} (\cO')}$ of $\pi$
to ${\cA} (\cO')$ is unitarily equivalent to $\pi_0|_{{\cA} (\cO')}$, i.e., 
\begin{equation}
\pi|_{{\cA} (\cO')}\cong \pi_0|_{{\cA} (\cO')}.
\end{equation}
This condition means that, if a local excitation specified by $\pi$
is localized in $\cO$, we cannot distinguish the excitation and the vacuum
in the causal complement $\cO'$.

For a representation $\pi$ of ${\cA} $ on $\cH$ satisfying the DHR selection criterion
in support with a bounded region $\cO$, there exists a unitary operator
$U:\cH_0\rightarrow\cH$ such that $\pi(A)U=U\pi_0(A)$ for all $A\in{\cA} (\cO')$.
We can define a representation $\pi'$ on $\cH_0$ by
\begin{equation}
\pi'(A)=U^* \pi(A)U
\end{equation}
for all $A\in{\cA} $. Then $\pi'$ satisfies
\begin{equation}
\pi'(A)=U^* \pi(A)U=U^* U\pi_0(A)=\pi_0(A)
\end{equation}
for all $A\in{\cA} (\cO')$, i.e.,
$\pi'|_{{\cA} (\cO')}= \pi_0|_{{\cA} (\cO')}$.
Therefore, we may only consider representations $\pi$ of ${\cA} $ on $\cH_0$ satisfying
\begin{equation}
\pi|_{{\cA} (\cO')}= \pi_0|_{{\cA} (\cO')}
\end{equation}
from the beginning, instead of representations of ${\cA} $ on different Hilbert spaces 
satisfying the DHR criterion.

In the DHR-DR theory, it is usually assumed that all physically relevant
factor representations satisfying the DHR selection criterion
are quasi-equivalent to irreducible ones. 
Here, we also assume this.
By this assumption and the categorical analysis by Doplicher and Roberts \cite{DR89b,DR90},
all representations satisfying the DHR selection criterion generate
atomic von Neumann algebras.
There then exists a normal conditional expectation
$\mathcal{E}_\pi:\B(\cH_0)\rightarrow\pi({\cA} ){''}$
for any representation $\pi$ on $\cH_0$ satisfying the DHR selection criterion.

Now, we shall enter into measurement theory in AQFT, the main subject of this section.
Measurement theory in AQFT has been discussed in several investigations.
For example, Doplicher \cite{Dop09} deepened
the relation between concepts of traditional measurement theory and those of AQFT.
Here, we develop the theory by applying the results given in the previous sections.

We have mentioned in section \ref{se:approximations} that local algebras 
${\cA} (\cO)$ are AFD and separable
under very general postulates, e.g., the Wightman axioms, nuclearity, 
and asymptotic scale invariance \cite{BDF87};
then we can assume that ${\cA} $ acts on $\mathcal{H}_0$ from the beginning, 
i.e., $\pi_0=\id$.
These postulates are strongly related to standard settings of quantum field theory and
are supposed to hold for typical models to which we can apply the DHR-DR theory.
In particular, the nuclearity condition is often assumed
since a quantum field modeled by a local net satisfying this condition possesses
a reasonable particle interpretation.
Furthermore, it is shown in \Cite{BW86} that the nuclearity condition implies 
the split property introduced as follows.

\begin{definition}[Split property]
Let $\{{\cA} (\cO)\}_{\cO\in\cK}$ be a family of W$^\ast$-algebras.
A pair $(\cO_1,\cO_2)\in\cK_\Subset$ is called a split pair
for $\{{\cA} (\cO)\}_{\cO\in\cK}$
if there exists a type $\mathrm{I}$ factor $\mathcal{N}$ such that
${\cA} (\cO_1)\subset\mathcal{N}\subset{\cA} (\cO_2)$.
We say that $\{{\cA} (\cO)\}_{\cO\in\cK}$ satisfies the split property
if every $(\cO_1,\cO_2)\in\cK_\Subset$ is a split pair
for $\{{\cA} (\cO)\}_{\cO\in\cK}$.
\end{definition}

If a local net $\{{\cA} (\cO)\}_{\cO\in\cK}$, which 
are von Neumann algebras on a Hilbert space,
satisfies the split property, then 
\beq
{\cA} (\cO_1)\vee{\cA} (\cO_2)
\cong{\cA} (\cO_1)\;\overline{\otimes}\;{\cA} (\cO_2)
\eeq
 holds
for any $\cO_1,\cO_2\in\cK$ such that
$\overline{\cO_1}\subsetneq(\cO_2)'$.

In the spirit of algebraic quantum theory, it is natural to consider 
the observable algebra $\pi({\cA} (\cO_1))$
only for a double cone $\cO_1$, 
where $\pi$ is a representation of ${\cA} $ on $\cH_0$
such that $\pi|_{{\cA} ((\cO_1)')}
= \pi_0|_{{\cA} ((\cO_1)')}$, or 
the case where local excitations exist only in a double cone $\cO_1$ 
for simplicity, since for every $\cO\in\cK$
there is a double cone $\cO_1$ such that $\cO\subset\cO_1$.
Then a measuring apparatus for the system 
specified by $\pi({\cA} (\cO_1))$, with the output variable 
taking values in a measurable space $(S,\mathcal{F})$, corresponds to
an instrument ${\cI} $ for $(\pi({\cA} (\cO_1)),S)$.
On the other hand, we have to accept an obvious fact that $\pi({\cA} (\cO_1))$ 
is just one of observable algebras of a quantum field described by the local net 
$\{\pi({\cA} (\cO))\}_{\cO\in\cK}$, and that 
any measurement carried out in a local region $\cO_1$ can also be consistently
described as a measurement taking place in any larger region $\cO$  including 
the original region $\cO_1$.  
Hence we demand that ${\cI} $ can be regarded as the restriction 
of an instrument $\widetilde{{\cI} }$ for $(\pi({\cA} ){''},S)$ 
satisfying some locality condition to $\pi({\cA} (\cO_1))$.
Then $\widetilde{{\cI} }$ and ${\cI} $ have to satisfy
$\widetilde{\cI}(A,\De)=\cI(A,\De)$
for all $\De\in\mathcal{F}$ and $A\in \pi({\cA} (\cO_1))$,
and we call $\widetilde{{\cI} }$ a global extension of ${\cI} $.
Here, we define a local instrument ${\cI} $ on $\pi({\cA} ){''}$  as follows:
We will see later that it is adequate to ensure the existence of a global extension 
$\widetilde{{\cI} }$ of an instrument ${\cI} $ for $(\pi({\cA} (\cO_1)),S)$.

\begin{definition}
Let $(S,\mathcal{F})$ be a measurable space.
Let $\pi$ be a representation of ${\cA} $ on $\cH_0$
such that $\pi|_{{\cA} ((\cO_0)')}
= \pi_0|_{{\cA} ((\cO_0)')}$ for a bounded region $\cO_0\in\cK$.
Let $(\cO_1,\cO_2)\in\cK_\Subset$ be a split pair for
$\{\pi({\cA} (\cO))\}_{\cO\in\cK}$.
A local instrument $\cI$ for $(\pi({\cA} ){''},S,\cO_1,\cO_2)$
is an instrument for $(\pi({\cA} ){''},S)$
satisfying
\begin{equation}\label{locality}
\cI(AB,\De)=\cI(A,\De)B
\end{equation}
for all $\De\in\mathcal{F}$, $A\in\pi({\cA} ){''}$
and $B\in\pi({\cA} ((\cO_2)')){''}$, and
\begin{equation}
\cI(A,\De)\in\pi({\cA} (\cO_1))
\end{equation}
for all $\De\in\mathcal{F}$ and $A\in\pi({\cA} (\cO_1))$.
\end{definition}

Let $\cO_1$ be a double cone and
 $\pi$ a representation of ${\cA} $ on $\cH_0$
such that $\pi|_{{\cA} ((\cO_1)')}=\pi_0|_{{\cA} ((\cO_1)')}$.
Suppose that $\{\pi_0({\cA} (\cO))\}_{\cO\in\cK}$
satisfies the split property. Let $(S,\mathcal{F})$ be a measurable space.
Let ${\cI} $ be an instrument for $(\pi({\cA} (\cO_1)),S)$.
By the above assumption, for every instrument ${\cI} $ for
$({\cA} (\cO_1),S)$,
there exists a local instrument $\widetilde{\cI}$ 
for $(\pi({\cA} ){''},S,\cO_1,\cO_2)$
such that $(\cO_1,\cO_2)\in\cK_{\Subset}$ and $\widetilde{\cI}(A,\De)=\cI(A,\De)$
for all $\De\in\mathcal{F}$ and $A\in\pi({\cA} (\cO_1))$.
Under the W$^\ast$-isomorphism $\iota:\pi({\cA} (\cO_1))
\vee\pi_0({\cA} (\cO_3))
\rightarrow\pi({\cA} (\cO))
\;\overline{\otimes}\;\pi_0({\cA} (\cO_3))$,
the restriction ${\cI} '$ of $\widetilde{\cI}$ to
$\pi({\cA} (\cO_1))\;\overline{\otimes}\;\pi_0({\cA} (\cO_3))$
must satisfy the following equality,
where $\cO_3\in\cK$ such that $\overline{\cO_1}\subsetneq(\cO_3)'$
and that there exists $\cO_4\in\cK$
satisfying $\overline{\cO_4}\subsetneq \cO_3$:
\begin{equation}
{\cI} '(A\otimes B,\De)={\cI} (A,\De)\otimes \id_{\pi_0({\cA} (\cO_3))}(B)
\end{equation}
for all $\De\in\mathcal{F}$, $A\in \pi({\cA} (\cO_1))$
and $B\in \pi_0({\cA} (\cO_3))$.
The split property implies the existence of a type I factor $\mathcal{N}$
such that $\pi_0({\cA} (\cO_4))\subset\mathcal{N}\subset\pi_0({\cA} (\cO_3))$.
Thus ${\cI} $ is completely positive since ${\cI} '(\cdot,\De)$ is positive
for all $\De\in\mathcal{F}$ and
$\mathcal{N}$ contains a weakly dense C$^\ast$-algebra
isomorphic to the algebra of compact operators on a separable infinite-dimensional Hilbert space.
Therefore, an instrument having a global extension is always completely positive.

\cut{\red{
Under the W$^\ast$-isomorphism $\iota:\pi({\cA} (\cO_1))
\vee\pi_0({\cA} ((\cO_2)'))
\rightarrow\pi({\cA} (\cO))
\;\overline{\otimes}\;\pi_0({\cA}( (\cO_2)'))$,
the restriction ${\cI} '$ of $\widetilde{\cI}$ to
$\pi({\cA} (\cO_1))\;\overline{\otimes}\;\pi_0({\cA}( (\cO_2)'))$
must satisfy the following equality,
where $\cO_3\in\cK$ such that $\overline{\cO_1}\subsetneq(\cO_3)'$
and that there exists $\cO_4\in\cK$
satisfying $\overline{\cO_4}\subsetneq \cO_3$:
\begin{equation}
{\cI} '(\Delta,A\otimes B)={\cI} (\Delta,A)\otimes id_{\pi_0({\cA} (\cO_3))}(B)
\end{equation}
for all $\Delta\in\mathcal{F}$, $A\in \pi({\cA} (\cO_1))$
and $B\in \pi_0({\cA} (\cO_3))$.
The split property implies the existence of a type I factor $\mathcal{N}$
such that $\pi_0({\cA} (\cO_4))\subset\mathcal{N}\subset\pi_0({\cA} (\cO_3))$.
Thus ${\cI} $ is completely positive since ${\cI} '(\Delta,\cdot)$ is positive
for all $\Delta\in\mathcal{F}$ and
$\mathcal{N}$ contains a weakly dense C$^\ast$-algebra
isomorphic to the algebra of compact operators on a separable infinite-dimensional Hilbert space.
Therefore, an instrument having a global extension is always completely positive.
}}

\begin{definition}[Minimal dilation]
For a CP instrument for $(\B(\cH),S)$,
the triplet $(\cK,E,V)$ is called a minimal dilation of $\cI$
if $\cK$ is a Hilbert space,
$E:\cF\rightarrow\B(\cK)$ is a spectral measure 
and $V$ is an isometry from $\cH$ into $\cH\otimes\cK$ such that
\begin{equation}
\cI(X,\De)=V^* (X\otimes E(\De))V
\end{equation}
for all $X\in\B(\cH)$, and
\begin{equation}
\cH\otimes\cK= \overline{\mathrm{span}}
\{(X\otimes E(\De))V\xi\;|\;\xi\in\cH,X\in\B(\cH), \De\in\cF\}.
\end{equation}
\end{definition}

Let $(S,\mathcal{F})$ be a measurable space, $\cO_1$ a double cone, and 
$\pi$ a representation of ${\cA} $ on $\cH_0$
such that $\pi|_{{\cA} ((\cO_1)')}
= \pi_0|_{{\cA} ((\cO_1)')}$.
Under the split property, we can prove that, for
every CP instrument $\cI$ for $(\pi({\cA} (\cO_1)),S)$ with the NEP,
 there exists a local CP instrument
$\widetilde{\cI}$ for $(\pi({\cA} ){''},S,\cO_1,\cO_2)$,
a global extension of $\cI$, as follows.

\begin{theorem}[]\label{localCP}
Suppose that $\{\pi_0({\cA} (\cO))\}_{\cO\in\cK}$ satisfies the
split property. Let $(S,\mathcal{F})$ be a measurable space, $\cO_1$ a double cone, and 
$\pi$ a representation of ${\cA} $ on $\cH_0$
such that $\pi|_{{\cA} ((\cO_1)')}
= \pi_0|_{{\cA} ((\cO_1)')}$.
For every CP instrument $\cI$ for $(\pi({\cA} (\cO_1)),S)$ with the NEP
and any double cone $\cO_2$ such that $(\cO_1,\cO_2)\in\cK_\Subset^{DC}$,
there exists a local CP instrument
$\widetilde{\cI}$ for $(\pi({\cA} ){''},S,\cO_1,\cO_2)$ with the NEP
such that $\widetilde{\cI}(A,\De)=\cI(A,\De)$
for all $A\in\pi({\cA} (\cO_1))$ and $\De\in\mathcal{F}$.
For the above local CP instrument $\widetilde{\cI}$ and
every CP instrument $\cI'$ for $(\B(\cH_0),S)$
such that $\cI'(A,\De)=\widetilde{\cI}(A,\De)$
for all $A\in\pi({\cA} ){''}$ and $\De\in\mathcal{F}$,
the minimal dilation $(\cK,E,V)$ of $\cI'$ satisfies the following intertwining relation:
\begin{equation}
VA=(A\otimes 1)V
\end{equation}
for all $A\in\pi({\cA} ((\cO_2)')){''}$.
\end{theorem}

\begin{proof}
Let $\cI$ be a CP instrument for $(\pi({\cA} (\cO_1)),S)$ with  the NEP.
Let $\cO_2$ be such that $(\cO_1,\cO_2)\in\cK_\Subset^{DC}$.
Since there exists a type I factor $\cN$ such that
$\pi({\cA} (\cO_1))\subset
\pi_0({\cA} (\cO_1))\subset\cN\subset\pi_0({\cA} (\cO_2))$,
it holds by \Cite{Buc74} that
\begin{equation}
\pi({\cA} (\cO_1))
\vee\pi_0({\cA} (\cO_2))^{'}
\cong\pi({\cA} (\cO_1))
\;\overline{\otimes}\;\pi_0({\cA} (\cO_2))^{'}.
\end{equation}
There then exists a CP instrument $\widetilde{\cI}_0$
for $(\pi({\cA} (\cO_1))
\;\overline{\otimes}\;\pi({\cA} ((\cO_2)'){''},S)$
with  the NEP such that
\begin{equation}
\widetilde{\cI}_0 (X\otimes Y,\De)=\cI(X,\De)\otimes Y
\end{equation}
for all $\De\in\mathcal{F}$, $X\in\pi({\cA} (\cO_1))$
and $Y\in\pi({\cA} ((\cO_2)')){''}$.
We identify the CP instrument $\widetilde{\cI}_0$
for $(\pi({\cA} (\cO_1))\;\overline{\otimes}\;\pi({\cA} ((\cO_2)')){''},S)$
with that for $(\pi({\cA} (\cO_1))\vee\pi({\cA} ((\cO_2)')){''},S)$ with the NEP.
Let $\widetilde{\cI}_{1}$ be a CP instrument for $(\B(\cH_0),S)$ obtained by
Theorem \ref{NEP} (iii) the NEP such that 
$\widetilde{\cI}_{1}(X,\De)=\widetilde{\cI}_0(X,\De)$
for all $\De\in\mathcal{F}$ and $X\in\pi({\cA} (\cO))
\vee\pi_0({\cA} ((\cO_2)')){''}$.
We define a CP instrument $\widetilde{\cI}_2$ for $(\B(\cH_0),S)$
by
\begin{equation}
\widetilde{\cI}_2(X,\De)=\mathcal{E}_\pi(\widetilde{\cI}_{1}(X,\De))
\end{equation}
for all $\De\in\mathcal{F}$ and $X\in\B(\cH_0)$,
and a CP instrument $\widetilde{\cI}$ for $(\pi({\cA} ){''},S)$ by
\begin{equation}
\widetilde{\cI}(X,\De)=\widetilde{\cI}_{2}(X,\De)
\end{equation}
for all $\De\in\mathcal{F}$ and $X\in\pi({\cA} ){''}$.
It is easily seen that $\widetilde{\cI}$ is a local CP instrument
 for $(\pi({\cA} ){''},S,\cO_1,\cO_2)$
such that $\widetilde{\cI}(A,\De)=\cI(A,\De)$
for all $\De\in\mathcal{F}$ and $A\in\pi({\cA} (\cO_1))$,
and that $\widetilde{\cI}(X,\De)=\widetilde{\cI}_2(X,\De)$
for all $\De\in\mathcal{F}$ and $X\in\pi({\cA} ){''}$.

Let $\cI'$ be a CP instrument for $(\B(\cH_0),S)$ such that $\cI'(A,\De)=\widetilde{\cI}(A,\De)$
for all $A\in\pi({\cA} ){''}$ and $\De\in\mathcal{F}$, and $(\cK,E,V)$ the minimal dilation of $\cI'$.
For every $A\in\pi({\cA} ((\cO_2)')){''}$, it holds that
\begin{align}
 &\;\;\;\; (VA-(A\otimes 1)V)^*(VA-(A\otimes 1)V)\nonumber \\
 &= A^* V^* VA-V^*(A^*\otimes 1)VA-
 A^* V^*(A\otimes 1)V+V^*(A^* A\otimes 1)V\nonumber \\
 &=A^* A-\cI'(A^*,S)A-
 A^*\cI'(A,S)+\cI'(A^*A,S) \nonumber\\
 &=A^* A-\cI(A^*,S)A-
 A^*\cI(A,S)+\cI(A^*A,S) \nonumber\\
 &=A^* A-A^* A-A^* A+A^* A=0.
\end{align}
We used here Eq.(\ref{locality}) to derive the last line from the fourth line.
Thus we have $VA=(A\otimes 1)V$ for all $A\in\pi({\cA} ((\cO_2)')){''}$.
\end{proof}

In the case where each ${\cA} (\cO)$ is injective and acts on $\cH_0$,
each $\pi({\cA} (\cO))$ is also injective
for every representation $\pi$ of ${\cA} $ on $\cH_0$
such that $\pi|_{{\cA} ((\cO_1)')}
= \pi_0|_{{\cA} ((\cO_1)')}$ for some double cone $\cO_1$.
Then we have $\mathrm{CPInst}_{\mathrm{AN}}(\pi({\cA} (\cO_1),S)=
\mathrm{CPInst}(\pi({\cA} (\cO_1)),S)$. 
By Theorem \ref{localCP} and by the previous discussions
in this section and in section \ref{se:approximations}, 
we established, in physically reasonable situations
for a quantum field modeled by a local net 
$\{\mathcal{A}(\mathcal{O})\}_{\mathcal{O}\in\cK}$,
that for every measuring apparatus $\bA(\bx)$ in a double cone $\mathcal{O}_1$, 
where excitations specified by a representation $\pi$ are localized,
there exists a CP instrument $\mathcal{I}$ defined on $\pi(\mathcal{A}(\mathcal{O}_1))$
that describes the statistical properties of $\bA(\bx)$,
and for any $\varepsilon>0$ there exists a measuring process
$\M$ on $\pi(\mathcal{A})^{\prime\prime}$ that 
defines a local CP instrument, which approximates $\mathcal{I}$ 
within the error limit $\varepsilon$.

We would like to emphasize that there is much room for improvement of Theorem \ref{localCP}.
For example, the following theorem holds as a variant of Theorem \ref{localCP}.
\begin{theorem}[]
Let $(S,\mathcal{F})$ be a measurable space, and 
$\pi$ a representation of ${\cA} $ on $\cH_0$
such that $\pi|_{{\cA} ((\cO_0)')}
= \pi_0|_{{\cA} ((\cO_0)')}$ for a bounded region $\cO_0$.
Let $\cO_1\in\cK$.
If there is a split pair $(\cO_1,\cO_2)\in\cK_\Subset$ for
$\{\pi({\cA} (\cO))\}_{\cO\in\cK}$,
for every CP instrument $\cI$ for $(\pi({\cA} (\cO_1)),S)$ with the NEP,
there exists a local CP instrument
$\widetilde{\cI}$ for $(\pi({\cA} ){''},S,\cO_1,\cO_2)$ with the NEP
such that $\widetilde{\cI}(A,\De)=\cI(A,\De)$
for all $A\in\pi({\cA} (\cO_1))$ and $\De\in\mathcal{F}$.
\end{theorem}
The proof is almost the same as that of Theorem \ref{localCP}.
The existence of a conditional expectation
$\mathcal{E}_\pi:\B(\mathcal{H}_0)\rightarrow\pi({\cA} ){''}$ 
due to the atomicity of $\pi(\cA)''$ is crucial also here.

We can consider another locality condition for CP instruments.
Strictly local CP instruments are defined as follows.

\begin{definition}[]
Let $(S,\mathcal{F})$ be a measurable space and 
$\pi$ a representation of ${\cA} $ on $\cH_0$
such that $\pi|_{{\cA} ((\cO_0)')}
= \pi_0|_{{\cA} ((\cO_0)')}$ for a bouned region $\cO_0\in\cK$.
Let $\cO_1\in\cK$.
A strictly local instrument $\cI$ for $(\pi({\cA} ){''},S,\cO_1)$
is an instrument for $(\pi({\cA} ){''},S)$
satisfying
\begin{equation}
\cI(AB,\De)=\cI(A,\De)B
\end{equation}
for all $\De\in\mathcal{F}$, $A\in\pi({\cA} ){''}$
and $B\in\pi({\cA} ((\cO_1)')){''}$, and
\begin{equation}
\cI(A,\De)\in\pi({\cA} (\cO_1))
\end{equation}
for all $\De\in\mathcal{F}$ and $A\in\pi({\cA} (\cO_1))$.
\end{definition}

If there is a split pair $(\cO_1,\cO_2)\in\cK_\Subset$ for
$\{\pi({\cA} (\cO))\}_{\cO\in\cK}$,
a strictly local instrument $\cI$ for $(\pi({\cA} ){''},S,\cO_1)$
is of course a local instrument $\cI$ for $(\pi({\cA} ){''},S,\cO_1,\cO_2)$.
This definition is a generalization of that of Halvorson \cite{Hal04JMP}
to general representations satisfying the DHR selection criterion.
The following proposition then holds:

\begin{proposition}[] \label{localCPinst}
Let $(S,\mathcal{F})$ be a measurable space and 
$\cO_1$ a double cone,
$\pi$ a representation of ${\cA} $ on $\cH_0$
such that $\pi|_{{\cA} ((\cO_1)')}
= \pi_0|_{{\cA} ((\cO_1)')}$.
Then every strictly local CP instrument $\cI$
for $(\pi({\cA} ){''},S,\cO_1)$ has  the NEP, and
the minimal dilation $(\cK,E,V)$ of the $\mathcal{E}_\pi$-canonical extension
$\widetilde{\cI}=\cE_{\pi}^*\cI$ of $\cI$ satisfies the following intertwining relation:
\begin{equation}
VA=(A\otimes 1)V
\end{equation}
for all $A\in\pi({\cA} ((\cO_1)')){''}$.
\end{proposition}
The proof of this proposition is similar to the last part of the proof of Theorem \ref{localCP}.

A typical example of strictly local CP instruments is the instrument 
$\cI$ for a von Neumann model of measurement of an observable $A$
affiliated to $\pi({\cA} (\cO_1))$ (cf. \Cite{93CA}), i.e., 
\beq
\cI(M,\De)=\int_{\De}\al(x1-A)^{*}M\,\al(x1-A)\,dx,
\eeq
where $\al\in L^2(\R)$, $\|\al\|_2=1$, $\De\in\cB(\R)$, and $M\in\pi({\cA} ){''}$.
Even if a CP instrument $\cI$ for $(\pi({\cA} (\cO_1)),S)$ has  the NEP,
there does not always exist a strictly local CP instrument $\widetilde{\cI}$ for
$(\pi({\cA} ){''},S,\cO_1)$
such that $\widetilde{\cI}(A,\De)=\cI(A,\De)$
for all $A\in\pi({\cA} (\cO_1))$ and $\De\in\mathcal{F}$.
A future work is to find a condition that
a CP instrument $\cI$ for $(\pi({\cA} (\cO_1)),S)$ has
a strictly local CP instrument $\widetilde{\cI}$ for
$(\pi({\cA} ){''},S,\cO_1)$ such that $\cI(A,\De)=\widetilde{\cI}(A,\De)$
for all $A\in\pi({\cA} (\cO))$ and $\De\in\mathcal{F}$.

In this section, we formulated local measurement on the basis of algebraic quantum field theory.
Our attempt is very natural because actual measurements are genuinely local.
On the other hand, there exist observables such as charges and as particle numbers, both of which are affiliated to global algebras but not to local algebras \cite{BDL86,DDFL87}.
This fact follows from origins of them. It is, however, known that we can actually measure them
in local regions. A typical example is photon counting measurement in quantum theoretical light detections,
which should be taken into account even when we treat light as quantum electro-magnetic field.
It is proved in \Cite{DDFL87} that there exists a net of self-adjoint operators affiliated to
W*-algebras of local observables which converges to a global charge under the split property.
In addition, in terms of nonstandard analysis \cite{CK90,Rob66},
we can describe ``infinitely large" local regions.
Therefore, it is expected that we can mathematically justify local measurements of 
global charges and of particle numbers.
This may be related to the reason why particle numbers should be treated as non-conserved quantities in wide situations in contrast to low-energy situations in which
their amounts are conserved.
Temperature and chemical potential are of the same kind in thermal situations.

In conclusion, it should be mentioned that any mathematical frameworks of quantum 
field theory are incomplete yet, and  algebraic quantum field theory is not an exception;
the most important and longstanding difficulty in AQFT is to show the existence of models on four-dimensional Minkowski space with non-trivial interactions. 
For further development in analysis of local measurements especially in concrete models 
in the near future,  considerable difficulties would be anticipated to be originated from those
difficulties, compared with quantum measurement theory for the systems with finite degrees
of freedom.
On the other hand, we expect that quantum measurement theory gives new insights 
into the description of interactions in AQFT.  We hope this paper stimulates readers to participate
the development in this attractive field of interplay between quantum measurement theory 
and AQFT.

\section*{Acknowledgement}
\label{se:acknowledgement}
K.O. would like to thank Prof. Izumi Ojima and Dr. Hayato Saigo 
for their warm encouragement and useful comments.
The authors thank the anonymous referee for useful comments.
This work was supported by the John Templeton Foundations, No.~35771
and the JSPS KAKENHI, No.~26247016.
M.O. acknowledges the support by the JSPS KAKENHI, No.~15K13456.

%

\end{document}